\documentclass[10pt,english,aps,prl,reprint,superscriptaddress,amssymb,amsfonts,longbibliography]{revtex4-1}
\usepackage[T1]{fontenc}
\usepackage[latin9]{inputenc}
\usepackage{mathrsfs}
\usepackage{mathtools}
\usepackage{amsmath}
\usepackage{amsthm}
\usepackage{bbm}
\usepackage{graphicx}
\usepackage{comment}

\usepackage{txfonts}

\makeatletter
\theoremstyle{definition}
\newtheorem{defn}{\protect\definitionname}
\theoremstyle{remark}
\newtheorem{rem}{\protect\remarkname}
\theoremstyle{definition}
\newtheorem{example}{\protect\examplename}
\theoremstyle{remark}
\newtheorem{claim}{\protect\claimname}

\usepackage[colorlinks,bookmarks=false,citecolor=blue,linkcolor=red,urlcolor=blue]{hyperref}
\usepackage{babel}

\renewcommand{\textemdash}{---}

\makeatother

\usepackage{babel}
\providecommand{\claimname}{Claim}
\providecommand{\definitionname}{Definition}
\providecommand{\examplename}{Example}
\providecommand{\remarkname}{Remark}

\begin{document}
	
	\title{Fracton Self-Statistics}
	
	\author{Hao Song}
	\affiliation{CAS Key Laboratory of Theoretical Physics, Institute of Theoretical Physics, Chinese Academy of Sciences, Beijing 100190, China}
	\affiliation{Department of Physics and Astronomy, McMaster University, Hamilton,	Ontario L8S 4M1, Canada}
	
	\author{Nathanan Tantivasadakarn}
	\affiliation{Walter Burke Institute for Theoretical Physics, California Institute of Technology, Pasadena, CA 91125, USA}
	\affiliation{Department of Physics, California Institute of Technology, Pasadena, CA 91125, USA}
	\affiliation{Department of Physics, Harvard University, Cambridge, MA 02138, USA}
	
	\author{Wilbur Shirley}
	
	\affiliation{School of Natural Sciences, Institute for Advanced Study, Princeton,
		NJ 08540, USA}
	
	\affiliation{Department of Physics, California Institute of Technology, Pasadena, CA 91125, USA}
	\affiliation{Institute for Quantum Information and Matter,
		California Institute of Technology, Pasadena, CA 91125, USA}
	
	\author{Michael Hermele}
	
	\affiliation{Department of Physics and Center for Theory of Quantum Matter, University
		of Colorado, Boulder, CO 80309, USA}
	
	\date{\today}
	\begin{abstract}
		Fracton order describes novel quantum phases of matter that host quasiparticles with restricted mobility, and thus lies beyond the existing paradigm of topological order. In particular, excitations that cannot move without creating multiple excitations are called fractons. Here we  address a fundamental open question---can the notion of self-exchange statistics be naturally defined for fractons, given their complete immobility as isolated excitations? Surprisingly, we demonstrate how fractons can be exchanged, and show that their self-statistics is a key part of the characterization of fracton orders. We derive general constraints satisfied by the fracton self-statistics in a large class of Abelian fracton orders. Finally, we show the existence of nontrivial fracton self-statistics in some twisted variants of the checkerboard model and Haah's code, establishing that these models are in distinct quantum phases as compared to their untwisted cousins.  
	\end{abstract}
	\maketitle
	
	\paragraph{Introduction.}
	Particle statistics is a fundamental aspect of quantum mechanics. While elementary particles that compose our universe must be either bosons or fermions due to the topological triviality of double exchanges in 3D space, 
	emergent quasiparticles in 2D quantum many-body systems can exhibit anyonic statistics~\citep{Arovas1984,Kitaev2006}, which  are crucial for characterizing conventional topological order. Recently, the theoretical discovery of fracton order in 3D \citep{Chamon2005,Bravyi2011,Haah2011,Yoshida_Fractal,Vijay2015,Vijay2016,Pretko2017}
	has revealed a new situation where quasiparticles lack their usual freedom to move in space, calling for a reexamination of the notion of statistics~\cite{Ma2017, TwistedFracton, Pai2019}. 
	
	Fracton systems have emerged as an active frontier of quantum physics~\cite{NandkishoreRev2019, PretkoRev2020}, attracting great interest from condensed matter, quantum information and quantum field theory viewpoints. Fracton order is defined by the emergence of quasiparticles with restricted mobility, including \emph{fractons}, which cannot move without splitting into more than one excitation. Single isolated fractons are thus immobile. Fracton models can also host excitations which are mobile only within planes or lines. Statistical processes involving or interpretable in terms of partially mobile excitations have been studied~\citep{Ma2017,Slagle2017,Vijay2017,TwistedFracton,CageNet,Shirley2019,Bulmash2019,Pai2019,You2020,ShirleyFF,TJV1,Shirley23}.
	Moreover, fractons can be non-Abelian in the sense of carrying protected topological degeneracy~\citep{Vijay2017,TwistedFracton,BulmashBarkeshli2019,PremWilliamson2019,WangXu19,WangXuYau19,WangYau19,AasenBulmashPremSlagleWilliamson20,Wen20,Wang20,WilliamsonCheng20,StephenGarre-RubioDuaWilliamson2020,TJV2,TuChang21}.
	Nevertheless, a fundamental question remains open: does a notion of self-exchange statistics make sense for fractons, given their complete immobility as isolated excitations?

	In this Letter, we provide a resolution to this puzzle. By allowing the fracton quasiparticle to split into multiple coordinated pieces, it is possible to prepare two well-separated realizations of the same fracton superselection sector. Such a pair of excitation patterns can be physically exchanged, giving rise to a fracton self-statistics. Our findings apply to both fracton phases of foliated \cite{ShirleyFoliated,GeneralizedFoliation} and fractal \cite{Haah2011,Yoshida_Fractal} nature. Furthermore, we point out instances where the self-statistics of fractons is in fact the only known statistical invariant that distinguishes between two fracton phases. We provide explicit examples by distinguishing twisted checkerboard models~\citep{TwistedFracton} and a twisted Haah's code~\citep{PhysRevResearch.2.023353} from their untwisted counterparts. Thus, we show that fracton self-statistics is a fundamental invariant needed to characterize fracton phases of matter.

	\paragraph{Foliated fractons.}To illustrate the principle, we start with the simplest relevant setting, in which all fractons $a$ are Abelian \cite{NAbelian} and satisfy the fusion constraint
	\begin{equation}
	a\times \overline{\prescript{\mathbf{t}_{\mu}}{}{a}} \times \prescript{\mathbf{t}_{\mu}+\mathbf{t}_{\nu}}{}{a}\times \overline{\prescript{\mathbf{t}_{\nu}}{}{a}}=1
	\label{eq:3foliatedfusion}
	\end{equation}
	for all $\mu,\nu\in\{x,y,z\}$ such that $\mu\neq\nu$, 
	where $\mathbf{t}_{\mu\,(\nu)}$ is the elementary lattice vector in the $\mu\,(\nu)$ direction, 
	$\prescript{\mathbf{t}}{}{a}$  
	denotes the analogue of $a$ at a $\mathbf{t}$-shifted position, and
	$\overline{\prescript{\mathbf{t}}{}{a}}$ is the antiparticle of $\prescript{\mathbf{t}}{}{a}$.
	This constraint guarantees the existence of rectangular~\cite{FoliationDirection} 
	membrane operators of arbitrary size that generate quadrupolar configurations of a given species $a$ at its corners. 
	Fractons satisfying the fusion constraint will be referred to as (Abelian) \textit{foliated}.

	A large body of models hosting foliated fractons are known in the literature, including the X-cube, checkerboard, and their many variants~\cite{Vijay2016,TwistedFracton,CageNet,you2020symmetric,ShirleySlagleChen20,DevakulShirleyWang20}. 
	Let us refer to the checkerboard model as a concrete example, for its twisted variants will clearly demonstrate the usage of fracton self-statistics. 
	
	\begin{figure}
		\includegraphics[width=1\columnwidth]{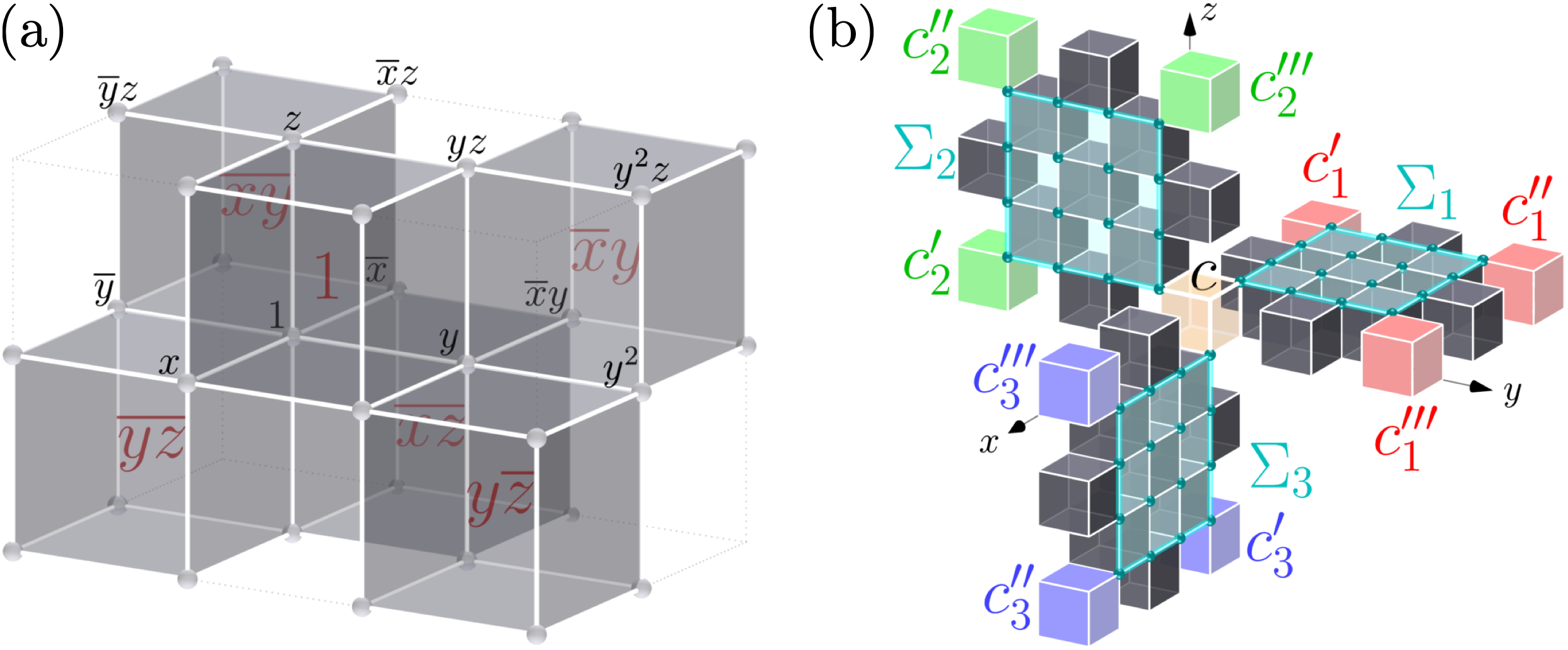}
		
		\caption{
			(a) 3D checkerboard lattice with vertices and gray cubes labeled by monomials $x^{j}y^{k}z^{l}$, where $\overline{x}\equiv x^{-1}$ etc.  
			(b) In the checkerboard model, excitation $A_{c}=-1$ (or \textbf{$B_{c}=-1$}) has \emph{fractional mobility}
			$\left\{ c\right\} \rightarrow \{ c_{i}^{\prime},c_{i}^{\prime\prime},c_{i}^{\prime\prime\prime} \} $,
			realizable by operator $M_{i}=\prod_{v\in\Sigma_{i}}Z_{v}$ (respectively,
			$M_{i}=\prod_{v\in\Sigma_{i}}X_{v}$) supported on rectangular membrane $\Sigma_{i}$, for $i=1,2,3$.}
		
		\label{fig:checkerboard}
	\end{figure}
	
	The \textit{checkerboard model}~\cite{Vijay2016} is defined on a 3D checkerboard lattice (Fig.~\ref{fig:checkerboard}a) with one qubit per vertex $v$. Its Hamiltonian 
	\begin{equation}
	H_{\mathrm{cb}}=-\sum_{c}\left(A_{c}+B_{c}\right) \label{eq:Hcb}
	\end{equation}
	is a summation over gray cubes $c$ in Fig.~\ref{fig:checkerboard}(a), where 
	\begin{equation}
	A_{c}\coloneqq\prod_{v\in c}X_{v},\qquad B_{c}\coloneqq\prod_{v\in c}Z_{v},
	\end{equation}
	are products of Pauli $X$ or $Z$ operators at the eight vertices of $c$. 
	This is an exactly solvable gapped model with spectrum labeled by simultaneous eigenvalues $\left\{ A_{c},B_{c}=\pm1\right\} $.
	
	An isolated excitation $A_{c}=-1$ exemplifies a foliated fracton. It
	can be ``moved'' at the expense of fractionalizing into more than
	one excitation, e.g., $\eta=\left\{ c\right\} \rightarrow\eta_{i}= \{ c_{i}^{\prime},c_{i}^{\prime\prime},c_{i}^{\prime\prime\prime} \} $ by a rectangular membrane operator; see Fig.~\ref{fig:checkerboard}(b).  Therefore, the excitation patterns
	$\eta_{1}$ (red), $\eta_{2}$ (green), $\eta_{3}$ (blue), and $\eta$ (orange) are all realizations
	of the same fracton superselection sector. 
	
	\paragraph{Self-statistics of foliated fractons.}Generically, a foliated fracton $a$ is characterized by a set of four self-statistical phases $\theta^{[xyz]}_a$, $\theta^{[x\overline{yz}]}_a$, $\theta^{[\overline{x}y\overline{z}]}_a$, and $\theta^{[\overline{xy}z]}_a$, 
	each corresponding to a ``windmill'' self-exchange process. 
	
	The process corresponding to $\theta^{[xyz]}_a$ is depicted in Fig.~\ref{fig:Process}. 
	It begins with an excited state with $a$ at the center of the windmill, in addition to a triplet of excitations denoted $\widehat{a}$ that belongs to the same superselection sector as $a$. The process proceeds with a sequence of six membrane operators (Fig.~\ref{fig:Windmill}a) whose total action exchanges $a$ with $\widehat{a}$, returning to the starting state in such a way that all arbitrary phases cancel. It can be regarded as a fractonic generalization of the T-shaped anyon exchange process \citep{Tprocess}.
	
	\begin{figure}
		\centering
		\includegraphics[width=.48\textwidth]{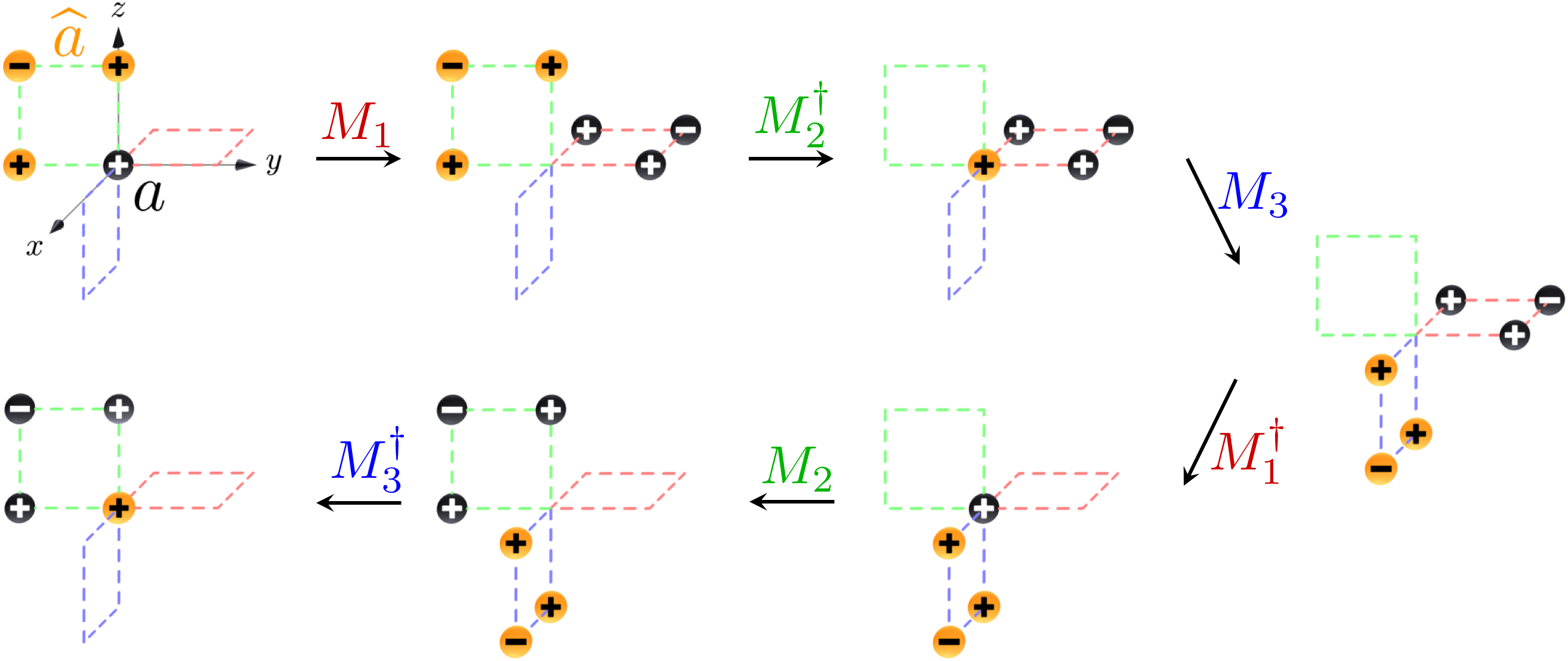}
		\caption{The $[xyz]$ windmill process. The starting state has an excitation $a$ at the center of the ``windmill'', along with three other excitations collectively called $\widehat{a}$ in the same superselection sector as $a$ following Eq.~\eqref{eq:3foliatedfusion}. The process involves three membrane operators $M_1$, $M_2$, $M_3$, and their inverses, successively moving the four excitations from the corners of the $yz$ square, to the corners of the $xy$ square, to the corners of the $zx$ square, and finally back to the original configuration. The process is designed such that the phase arbitrariness in the choice of $M_i$ is precisely cancelled by the action of $M_i^\dagger$. Therefore, the universal statistical phase is well-defined by $\theta_a^{[xyz]}=M_3^\dagger M_2 M_1^\dagger M_3 M_2^\dagger M_1$.}
		\label{fig:Process}
	\end{figure}
	
	The processes for $\theta^{[x\overline{yz}]}_a$, $\theta^{[\overline{x}y\overline{z}]}_a$, and $\theta^{[\overline{xy}z]}_a$ are defined analogously, but along windmills 
	related to $[xyz]$ by a $180^\circ$-rotation about the $x$, $y$, and $z$ axes, respectively. For instance, the $[x\overline{yz}]$ process involves the membrane operators located as in Fig.~\ref{fig:Windmill}(b). 
	The notation $[\mu_{1}\mu_{2}\mu_{3}]$ of three directions $\mu_{i}$ refers to a windmill made of three blades $K_{i}=\mathrm{cone}(-\boldsymbol{\mu}_{i},\boldsymbol{\mu}_{i+1})\equiv\{-\alpha\boldsymbol{\mu}_{i}+\beta\boldsymbol{\mu}_{i+1}|\alpha,\beta\geq0\}$ 
	for $i=1,2,3$, where $\mu_{4}\equiv\mu_{1}$. Each overlined direction indicates its opposite (e.g., $\overline{\mathbf{x}}=-\mathbf{x}$). 
	
	Although more windmill processes can be considered, they yield \emph{no} new self-statistical phases beyond the four already defined. 
	Any two inversion-related windmills (e.g., $\left[xyz\right]$ and $\left[\overline{xyz}\right]$ in Fig.~\ref{fig:Windmill}(a) and (c)) specify the same self-statistics. The reason is demonstrated in Fig.~\ref{fig:Windmill}(d): membrane operators for $\left[xyz\right]$ and $\left[\overline{xyz}\right]$ can be related by a deformation~\cite{RA2}. Consequently, despite eight possible windmill choices (see SM~\cite{SM}), only four self-statistics need to be specified for foliated fractons.  
	
	One might expect that $\theta^{[xyz]}_a$, $\theta^{[x\overline{yz}]}_a$, $\theta^{[\overline{x}y\overline{z}]}_a$, and $\theta^{[\overline{xy}z]}_a$ are independent. To the contrary, they are subject to a constraint
	\begin{equation}
	\theta_a^{[xyz]}\theta_a^{[x\overline{yz}]}\theta_a^{[\overline{x}y\overline{z}]}\theta_a^{[\overline{xy}z]}=1,
	\label{eq:FourfoldConstraint}
	\end{equation}
	leaving only three of them independent in general.

	\begin{figure}
		\centering
		\includegraphics[width=1\columnwidth]{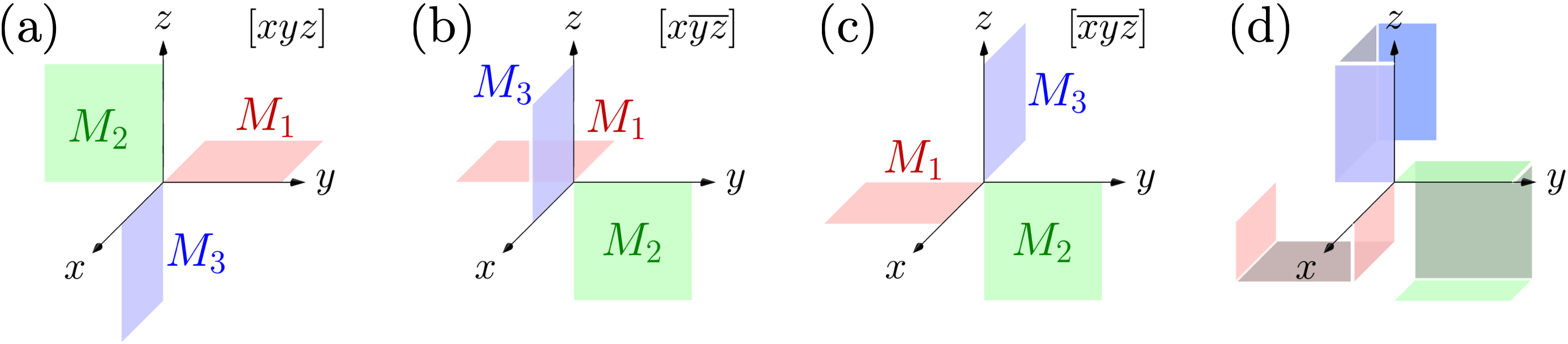}
		\caption{Membrane operators comprising the (a) $[xyz]$, (b) $[x\overline{yz}]$, and (c) $[\overline{xyz}]$ windmill processes. (d) The membrane operators for the $[\overline{xyz}]$ process are smoothly deformed such that, near the origin, they coincide with those of the $[xyz]$ process. This proves  $\theta_{a}^{[\overline{xyz}]} \equiv \theta_{a}^{[xyz]}$.
		}
		\label{fig:Windmill}
	\end{figure}

	This constraint is most naturally derived by utilizing a quantity $S_{ab}^\mu$ for $\mu=x,y,z$, defined as the mutual braiding statistics between dipoles $a \times \overline{\prescript{l\mathbf{t}_{\mu}}{}{a}}$ and $b\times \overline{\prescript{-l\mathbf{t}_{\mu}}{}{b}}$ in the large $l$ limit. The dipoles are planons (i.e., quasiparticles mobile in two directions). The braiding direction is fixed by $\mu$ via the right hand rule. See Fig.~\ref{fig:DoubleExchange}(a).
	
	\begin{figure}
		\centering
		\includegraphics[width=1\columnwidth]{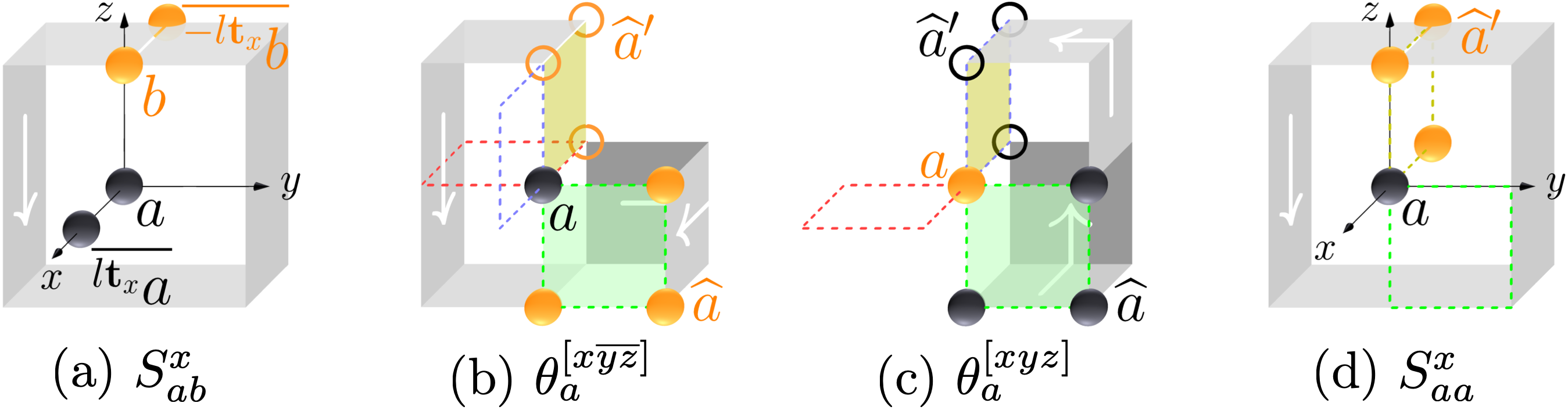}
		\caption{Graphic proof of $\theta_{a}^{[x\overline{yz}]}\theta_{a}^{[xyz]}=S_{aa}^{x}$. The white arrows denote the direction of braiding and exchange processes. (a) Definition of $S_{ab}^{x}$. (b) The $[x\overline{yz}]$ process (dotted windmill) is deformable into one realized in three steps $a \rightarrow \widehat{a}'$,  $\widehat{a}\rightarrow a$, and $\widehat{a}' \rightarrow \widehat{a}$ using operators supported on
			the olive, green, and gray areas. The intermediate state $\widehat{a}'$ consists of excitations at the three circles. (c) A process which is equivalent to the $[\overline{xyz}]$ process and hence produces statistics $\theta_{a}^{[\overline{xyz}]} \equiv \theta_{a}^{[xyz]}$. (d) A process that braids part of $\widehat{a}'$, along on the gray ribbon, around $a$. The statistical phase due to the presence of $a$ is $S_{aa}^x$.}
		\label{fig:DoubleExchange}
	\end{figure}

	A proof of Eq.~\eqref{eq:FourfoldConstraint} is as follows.
	If $a$ is exchanged twice with $\widehat{a}$, both sets of excitations return to their original position. The total process is smoothly deformable into one where $a$ is stationary while $\widehat{a}$ braids around $a$.
	For instance, we can deform the $[x\overline{yz}]$ windmill process into one along the cyclic ``path'' in Fig.~\ref{fig:DoubleExchange}(b). 
	Similarly, the $[\overline{xyz}]$ process (which produces statistics $\theta_{a}^{[\overline{xyz}]} \equiv \theta_{a}^{[xyz]}$) is deformable into the one depicted in Fig.~\ref{fig:DoubleExchange}(c).
	If the two deformed exchanges are started and ended with the intermediate state containing excitations $a$ and $\widehat{a}'$, their composite gives the process in Fig.~\ref{fig:DoubleExchange}(d), implying
	\begin{align}
	\theta^{[xyz]}_a\theta^{[x\overline{yz}]}_a&=S^x_{aa}.\label{eq:DoubleExchange2}
	\end{align}
	Armed with this relation, we can now prove Eq.~(\ref{eq:FourfoldConstraint}) by a $180\textdegree$-rotation of Eq.~(\ref{eq:DoubleExchange2}) about the $y$ axis to obtain $\theta^{[\overline{x}y\overline{z}]}_a\theta^{[\overline{xy}z]}_a=(S^x_{aa})^*$, and multiplying it with Eq.~(\ref{eq:DoubleExchange2}).

	The mutual statistics also appear in the following formula for the self-statistics of a fusion product of two fractons, and analogous formulas due to cubic symmetry:
	\begin{equation}
	\theta_{a\times b}^{[xyz]}=\theta_{a}^{[xyz]}\theta_{b}^{[xyz]}S_{ab}^xS_{ab}^yS_{ab}^z.
	\label{eq:FusionExchange}
	\end{equation}
	See
	SM~\cite{SM} for a proof. 
	This relation implies
	\begin{equation}  S_{ab}^xS_{ab}^yS_{ab}^z=S_{ba}^{x}S_{ba}^{y}S_{ba}^{z}.
	\end{equation}
	It is interesting to note that Eqs.~(\ref{eq:DoubleExchange2}) and (\ref{eq:FusionExchange}) generalize the constraints $\theta_a^2=S_{aa}$ and $\theta_{a\times b}=\theta_a\theta_bS_{ab}$ of 2D Abelian topological orders, where $\theta_a$ is the topological spin and $S$ the topological $S$-matrix. 
	For an Abelian planon $a$ satisfying the foliation condition Eq.~\eqref{eq:3foliatedfusion}, analogous windmill processes are reducible into 2D braidings and the above discussions reduce to these familiar 2D equations.  
	
	Now assume a foliated fracton satisfies $a^N=1$. We show its self-statistics being constrained to \textit{discrete} values for use in distinguishing fracton orders. Note $S_{aa}^xS_{aa}^yS_{aa}^z=(\theta_a^{[xyz]})^2$ by virtue of Eqs.~\eqref{eq:FourfoldConstraint} and \eqref{eq:DoubleExchange2}. 
	Thus, since $(S^\mu_{aa})^N = S^\mu_{a^Na} = 1$, we have $(\theta_a^{[xyz]})^{2N}=1$.
	Moreover, applying Eq.~(\ref{eq:FusionExchange}) recursively gives
	$(\theta_a^{[xyz]})^{N^2} = \theta_{a^N}^{[xyz]}=1$. Together, these imply the self-statistics of $a$ being multiples of $e^{2 \pi i/(N \text{gcd}(N,2))}$ in analogy to anyons in 2D.
	
	\paragraph{Semionic fractons in twisted checkerboard models.}A major application of fracton self-statistics is to distinguish the quantum phase of the checkerboard model $H_{\mathrm{cb}}$ from its twisted variants introduced in Ref.~\citep{TwistedFracton}. 
	To illustrate, we consider seven twisted models, denoted $H_{\mathrm{cb}}^{x}$,
	$H_{\mathrm{cb}}^{y}$, $H_{\mathrm{cb}}^{z}$, $H_{\mathrm{cb}}^{xy}$,
	$H_{\mathrm{cb}}^{yz}$, $H_{\mathrm{cb}}^{zx}$, and $H_{\mathrm{cb}}^{xyz}$ below.
	Together with $H_{\mathrm{cb}}$, we will show that the eight models fall into two quantum phases, distinguishable by the presence or absence of semionic fracton self-statistics. 
	Explicit construction of paths connecting models with identical fracton self-statistics is given in SM~\cite{SM}.

	First, in $H_{\mathrm{cb}}$ (Eq.~\eqref{eq:Hcb}), all excitations (including fractons)
	exhibit either bosonic ($+1$) or fermionic ($-1$) statistics. This is because all statistical processes are realizable by tensor products of Pauli operators which only commute or anticommute with each other.

	In contrast, $H_{\mathrm{cb}}^{x}$ represents a new phase allowing \emph{semionic} ($\pm i$) fracton self-statistics. Instead of using the formalism in Ref.~\citep{TwistedFracton},
	we specify this model using a non-Pauli stabilizer Hamiltonian
	\begin{align}
	H_{\mathrm{cb}}^{x} & =-\sum_{c}\left(A_{c}^{x}+B_{c}\right)\label{eq:Hx}
	\end{align}
	obtained by replacing $A_c$ 
	in the untwisted model Eq.~\eqref{eq:Hcb} with a modified term $A_{c}^{x}$,
	to have a convenient description of excitations with $\left(A_{c}^{x}\right)^{2}=1$ and the full spectrum labeled by simultaneous eigenvalues $\left\{ A_{c}^{x},B_{c}=\pm1\right\} $, where $x$ refers to twisting being associated with $x$-edges.
	Explicitly, 
	we label vertices and cubes by monomials as
	in Fig.~\ref{fig:checkerboard}(a) and denote finite sets of vertices by polynomials with $\mathbb{Z}_2=\{0,1\}$ coefficients~\cite{Haah2013}. In this notation, 
	\begin{equation}
	A_{c}^{x}\coloneqq A_{c}\phi_{\left(1+x\right)\overline{x}c}\phi_{\left(1+x\right)xc} \label{eq:A_x}
	\end{equation}
	according to the construction described in SM~\citep{SM},
	where $\ell=\left(1+x\right)\overline{x}c$ and $\ell=\left(1+x\right)xc$ denote vertex pairs that are ends of $x$-edges,
	and
	\begin{align}
	\phi_{\ell} & \coloneqq\left(-1\right)^{n_{\ell\overline{y}}^{-}n_{\ell}^{-}+n_{\ell\overline{y}z}^{-}n_{\ell z}^{-}+n_{\ell}^{-}n_{\ell y}^{+}+n_{\ell z}^{-}n_{\ell yz}^{+}+n_{\ell y}^{+}n_{\ell y^{2}}^{-}+n_{\ell yz}^{+}n_{\ell y^{2}z}^{-}}\nonumber \\
	& \phantom{=}\cdot\left(-1\right)^{n_{\ell yz\left(1+y\right)}^{-}n_{\ell y\left(1+y\right)\left(1+z\right)}^{-}}\cdot i^{-n_{\ell\left(1+\overline{y}\right)\left(1+z\right)}^{-}} \label{eq:a}
	\end{align}
	is a Dijkgraaf-Witten twisting factor, with the shorthand
	\begin{equation}
	Z_{\kappa}\coloneqq\prod_{v\in\kappa}Z_{v},\qquad n_{\kappa}^{\pm}\coloneqq\frac{1}{2}\left(1\pm Z_{\kappa}\right) \text{,}
	\end{equation}
	for $\kappa$ any finite set of vertices.

	In $H_{\mathrm{cb}}^{x}$, one example of semionic fracton is a $B_{c}=-1$ excitation (denoted
	$m_{x}$ below), which has
	\begin{equation}
	\theta_{m_{x}}^{[xyz]}=\theta_{m_{x}}^{[x\overline{yz}]}=i\quad\mathrm{and}\quad\theta_{m_{x}}^{[\overline{xy}z]}=\theta_{m_{x}}^{[\overline{x}y\overline{z}]}=-i.
	\end{equation}
	Two derivations of the statistics are given in SM~\citep{SM}. 
	In one, we construct a modified $X$ operator that explicitly generates the statistical processes for $B$ excitations. The modification of $X$ is required to ensure no $A_{c}^{x}$ terms flipped,
	and results in the above semionic self-statistics.

	We emphasize that in $H_{\mathrm{cb}}^{x}$, exotic self-statistics ($\theta\neq\pm1$) are exclusive to fractons. 
	Ref.~\citep{TwistedFracton} reported that non-fractonic excitations in $H_{\mathrm{cb}}^{x}$ exhibit only bosonic or fermionic statistics. This implies that $H_{\mathrm{cb}}^{x}$ cannot be a tensor product of $H_\mathrm{cb}$ and 2D anyon models containing semions. Therefore, the fact that only fracton self-statistics can distinguish the two models highlights the novelty of $H_{\mathrm{cb}}^{x}$ as a distinct phase of matter. We refer to the phase of $H_{\mathrm{cb}}^{x}$ as a \emph{semionic fracton order}, as characterized by the presence of semionic statistics for only the fracton excitations.

	The remaining six models are constructed similarly to $H_{\mathrm{cb}}^{x}$.
	In $H_{\mathrm{cb}}^{x}$, the twisting factor $\phi_{\left(1+x\right)\overline{x}c}\phi_{\left(1+x\right)xc}$ in Eq.~\eqref{eq:A_x} is linked to $x$-edges.  
	Its analogue associated with $y$-edges ($z$-edges) specifies  $H_{\mathrm{cb}}^{y}$ ($H_{\mathrm{cb}}^{z}$). 
	Moreover, twisting can be applied to more than one direction simultaneously; for example, 
	$H_{\mathrm{cb}}^{xy}$ has twisting made along both $x$-edges and $y$-edges.

	Remarkably, despite the six models having different ground states, we discover that: 
	(1) 
	$H_{\mathrm{cb}}^{y}$,  $H_{\mathrm{cb}}^{z}$, and $H_{\mathrm{cb}}^{xyz}$ represent the same semionic fracton phase as $H_{\mathrm{cb}}^{x}$, 
	while (2) 
	$H_{\mathrm{cb}}^{xy}$, $H_{\mathrm{cb}}^{yz}$, and $H_{\mathrm{cb}}^{zx}$
	fall within the phase of  
	$H_{\mathrm{cb}}$.  
	Let us first demonstrate how fracton self-statistics are matched between $H_{\mathrm{cb}}^{xy}$ and $H_{\mathrm{cb}}$. 
	In $H_{\mathrm{cb}}^{xy}$,  excitation $B_c =-1$ (denoted $m_{xy}$) is a fracton with 
	\begin{equation}
	\begin{gathered}\theta_{m_{xy}}^{[xyz]}=i\cdot i=-1,\;\;\quad\theta_{m_{xy}}^{[\overline{xy}z]}=\left(-i\right)^{2}=-1,\\
	\theta_{m_{xy}}^{[x\overline{yz}]}=i\cdot\left(-i\right)=1,\quad\theta_{m_{xy}}^{[\overline{x}y\overline{z}]}=\left(-i\right)\cdot i=1,
	\end{gathered}
	\end{equation}
	where two twistings cause a cancellation in semionic character. 
	Further, combining $m_{xy}$ with an $A$ excitation at relative position $\overline{x}y$, denoted $\prescript{\overline{x}y}{}{e}$, yields a fracton $\prescript{\overline{x}y}{}{e}\times m_{xy}$ with purely bosonic self-statistics, which can be seen via Eq.~\eqref{eq:FusionExchange} and its analogues.

	Based on this observation, we indeed find an exact local unitary transformation relating the ground states of $H_\mathrm{cb}^{xy}$ and $H_\mathrm{cb}$, rigorously confirming they represent the same phase (see SM~\cite{SM}). 
	Other phase identifications in the classification can be proven analogously.

	\paragraph{Self-statistics of fractal fractons.}
	The notion of self-statistics extends to non-foliated fractons \cite{GF}.
	We demonstrate this with Haah's code~\cite{Haah2011}
	\begin{equation}
	H_{\mathrm{Haah}}=-\sum_{\lambda\in\Lambda}\left(A_{\lambda}+B_{\lambda}\right), \label{eq:haah}
	\end{equation}
	an exactly solvable model defined on a cubic lattice with two qubits per vertex. Here, 
	$\Lambda=\{x^{i}y^{j}z^{k}\}$ represents lattice vectors $\left(i,j,k\right)\in\mathbb{Z}^{3}$ in monomial form.  The $A\left(B\right)$ terms are translations of the representative $A_{1}\left(B_{1}\right)$ at the origin given in Fig.~\ref{fig:haah-1}(a). 
	Each $A_{\lambda}\left(B_{\lambda}\right)$ is a product of eight Pauli $X$'s ($Z$'s). 
	With collections of translationally related objects represented as sums of $\Lambda$'s elements,
	we can describe $A_{\lambda}$ and $B_{\lambda}$ using Laurent polynomials with $\mathbb{Z}_2=\{0,1\}$ coefficients~\cite{Haah2013}:
	\begin{align}
	A_\lambda &=\lambda \cdot (
	\overline{f}_1, 
	\overline{f}_2, 
	0,
	0),
	&B_\lambda &= \lambda \cdot (
	0,
	0,
	f_2,
	f_1),\\
	f_{1} & =1+x+y+z,& f_{2}&=1+xy+yz+zx,\label{eq:f}\\
	\overline{f}_{1} & =1+\overline{x}+\overline{y}+\overline{z},&\overline{f}_{2}&=1+\overline{xy}+\overline{yz}+\overline{zx}.\label{eq:f_inversion}
	\end{align}
	where 
	the first (last) two components of $A_\lambda$ and $B_\lambda$
	locate Pauli $X$'s ($Z$'s) for the two qubit species. The bar denotes spatial inversion: $x\rightarrow \overline{x} \equiv x^{-1}$ etc.

	Excitations can also be described by polynomials. Applying a Pauli $Z$ to the first (or second) qubit at the origin excites $A$-terms in the pattern $f_1$ (respectively, $f_2$). Interestingly, one may  flip $A$-terms purely in the $yz$ plane~\citep{Yoshida_Fractal} by noting 
	\begin{align}
	(y+z)f_1 + f_2 =   1+y+y^2 + z+yz+z^2 \eqqcolon g.
	\end{align}
	Consider planar fractional moves for visual clarity. The $yz$-planar ones are generated by $g$, allowing $A$ excitations
	to travel arbitrarily long distances toward each of the \emph{conic} directions $K_1, K_2$, and $K_3$ in Fig.~\ref{fig:haah-1}(a).  
	Explicitly, for $l=2^n$, one has $1+g^{l}=y^l +y^{2l} +z^l +z^l y^l +z^{2l}$ due to the $\mathbb{Z}_2$ setting of the model. Accordingly, 
	$1\rightarrow 1+g^{l}$ provides an instance of pushing an $A$-term excitation by at least $l$ distance toward $K_1$. 
	It is realizable by fractal-shaped operator 
	$g^{l-1}(0,0,y+z,1)$, reflecting the excitation being a fracton of fractal nature. See Fig.~\ref{fig:haah-1}(b). We call each $K_i$ a \textit{mobility cone} for $A$ excitations, as defined in SM~\cite{SM}. 
	The description of $B$ excitations are analogous but with spatial directions inverted. 
	
	Based on mobility cones, we categorize fractons of $H_{\mathrm{Haah}}$ into three types---$A$, $B$, and mixed---and define their windmill self-statistical processes. 
	Type-$A$ (type-$B$) are fractons with the mobility cones $K_i$ (respectively, $-K_i$) for $i=1,2,3$ shown in Fig.~\ref{fig:haah-1}(a). 
	The mixed are bound states of type-$A$ and type-$B$; 
	they cannot be moved along any individual cone among $K_i$'s or $-K_i$'s. 
	Self-statistics is definable using the ``windmill'' made of 
	mobility cones. See Figs.~\ref{fig:haah-1}(a) and (c). 
	In $H_\mathrm{Haah}$, non-mixed (i.e., type-$A$ or type-$B$)
	fractons exhibit purely bosonic self-statistics, since 
	only one type of Pauli is involved.
	
	\begin{figure}[t!]
		\includegraphics[width=1\columnwidth]{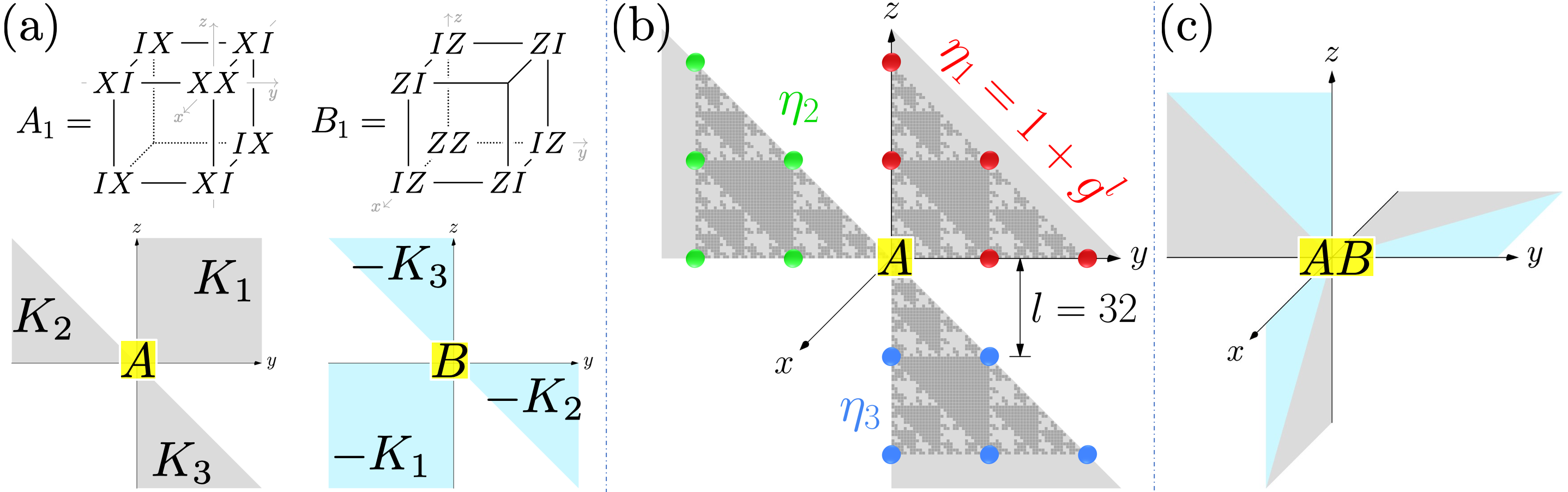}
		
		\caption{Fracton's mobility in the Haah's code. 
			(a) 
			Top: definition of $A_1$ and $B_1$. They are products of eight Pauli's. Identity operators $I$ are omitted when possible. 
			Bottom: mobility cones (on the $yz$-plane) for $A$ and $B$ excitations.
			(b)
			Fractional moves $1\rightarrow \eta_{i}$ of an $A$ excitation are realized by operators of fractal support. Gray square dots represent  operator $(0,0,y+z,1)$ and its translations.
			(c) 
			A windmill for a composite of type-$A$ and type-$B$ fractons. }
		
		\label{fig:haah-1}
	\end{figure}
	
	\paragraph{Fermionic type-$A$ fractons in a twisted Haah's code.}To further illustrate the usage of fracton self-statistics, consider a gauge-theoretic variant of Haah's code defined by applying $H_\mathrm{Haah}$ to a Hilbert space that binds a \emph{fermionic} mode $\psi_\lambda$ to $A_\lambda$ via Gauss's law $-i\gamma_{\lambda}\tilde{\gamma}_{\lambda}A_{\lambda}=1$,
	where $\gamma_{\lambda}\coloneqq\psi_{\lambda}+\psi_{\lambda}^{\dagger}$
	and
	$\tilde{\gamma}_{\lambda}\coloneqq\frac{1}{i}(\psi_{\lambda}-\psi_{\lambda}^{\dagger})$ are Majorana operators. 
	As detailed in SM~\cite{SM},
	the gauge theory emerges from a spin model $H_{\mathrm{Haah}}^F$, namely, the twisted Haah's code proposed in Ref.~\citep{PhysRevResearch.2.023353}.
	
	Fracton self-statistics enables us to settle the unresolved question of whether $H_{\mathrm{Haah}}^F$ represents a distinct fracton order from the original Haah's code $H_{\mathrm{Haah}}$. 
	The expectation that $A$ excitation becomes fermionic 
	is now definable and provable via windmill processes. 
	The operator creating $A$-excitations is modified to $Z_{\sigma} c_{\sigma}$ due to gauge invariance,
	where $Z_\sigma$ denotes Pauli $Z$ on qubit $\sigma$ while 
	$c_{\sigma}$ denotes a product of $\gamma_\lambda$'s that are associated with the $Z_\sigma$-flipped $A$ terms. 
	Still, one may wonder  
	whether it is possible to compensate the statistics change by attaching $B$ excitations to $A$. Indeed, this is the case for the 2D toric code, and the checkerboard model, which we have shown above. However, it is not allowed here because attaching type-$B$ fractons alters the mobility of $A$. Thus, the presence of fermionic type-$A$ fractons distinguishes $H_{\mathrm{Haah}}^{F}$ from  $H_{\mathrm{Haah}}$.
	See also SM~\cite{SM} for the discreteness of this self-statistics, which confirms the phase distinction.
	
	\paragraph{Conclusions.}We have shown that it is possible to exchange two realizations of a fracton superselection sector via its fractional mobility. The notion of self-statistics for fractons can thus be introduced, which is essential in characterizing fracton orders. As applications, we studied a family of twisted checkerboard models and a twisted Haah's
	code, from which we revealed a novel phase of foliated nature\textemdash what we call a semionic fracton order\textemdash and a new fractal-type
	order characterized by emergent fermionic fractons.
	Our work marks a crucial step towards a full ``algebraic theory of fractons'' yet to be developed. 
	
	\begin{acknowledgments}
		We thank Sheng-Jie Huang, Juven Wang and especially Ashvin Vishwanath for helpful discussions. The authors are grateful to the Banff International Research Station, where this work began in 2020 at the workshop ``Fractons and Beyond.'' HS also acknowledges discussions with Sung-Sik Lee. HS has been supported by the Natural Sciences and Engineering Research Council of Canada and the National Natural Science Foundation of China (Grant No.~12047503). NT is supported by the Walter Burke Institute for Theoretical Physics at Caltech. The work of MH is supported by the U.S. Department of Energy, Office of Science, Basic Energy Sciences (BES) under Award number DE-SC0014415. WS is supported by the Simons Collaboration on Ultra-Quantum Matter (UQM), which is a grant from the Simons Foundation (651444). The work of NT, WS and MH also benefited from meetings of the UQM Simons Collaboration supported by Simons Foundation grant number 618615.
	\end{acknowledgments}

	\bibliography{fss}

\begin{thebibliography}{57}
\bibitem [{NAb()}]{NAbelian}%
  \BibitemOpen
  \href@noop {} {}\bibinfo {note} {Non-Abelian generalizations will be studied
  in the future.}\BibitemShut {Stop}%
\bibitem [{Fol()}]{FoliationDirection}%
  \BibitemOpen
  \href@noop {} {}\bibinfo {note} {Generally, $\left\{ x,y,z\right\} $ may be
  oblique axes and accordingly membranes are parallelogrammatic. However, for
  convenience, we assume $\left\{ x,y,z\right\} $ to be orthogonal in this
  paper.}\BibitemShut {Stop}%
\bibitem [{RA2()}]{RA2}%
  \BibitemOpen
  \href@noop {} {}\bibinfo {note} {They are matched
  $(M_{1},M_{2},M_{3})_{[\overline{xyz}]}\sim(M_{3},M_{1},M_{2})_{[xyz]}$ up to
  a permutation, leaving self-statistics unchanged.}\BibitemShut {Stop}%
\bibitem [{GF()}]{GF}%
  \BibitemOpen
  \href@noop {} {}\bibinfo {note} {Different fractional mobility structures
  impose distinct constraints on fracton statistics. The interplay will be
  interesting for future studies.}\BibitemShut {Stop}%
\end{thebibliography}

\onecolumngrid 
\clearpage
\makeatletter 
\begin{center}   
	\textbf{\large Supplementary Material for ``Fracton Self-Statistics''}\\
	[1em]
	Hao Song, Nathanan Tantivasadakarn, Wilbur Shirley, and Michael Hermele  
	\thispagestyle{titlepage} 
\end{center} 	
\setcounter{equation}{0} 
\setcounter{figure}{0} 
\setcounter{table}{0} 
\setcounter{page}{1} 
\setcounter{section}{0} 
\renewcommand{\theequation}{S\arabic{equation}} 
\renewcommand{\thesection}{S.\Roman{section}}
\renewcommand{\thetable}{S\arabic{table}} 
\renewcommand\thefigure{S\arabic{figure}} 
\renewcommand{\theHtable}{Supplement.\thetable} 
\renewcommand{\theHfigure}{Supplement.\thefigure}

This supplementary material contains: (\ref{sec:counting}) the counting of windmills for foliated fractons, (\ref{sec:P})
a proof of Eq.~\eqref{eq:FusionExchange}, (\ref{sec:TGT}) a non-Pauli stabilizer formulation
of 2D double semion model (as a preliminary for twisted checkerboard
models), (\ref{sec:TwistedChecker}) details of twisted checkerboard
models and their classification, (\ref{sec:mobility-cones}) a mathematical
treatment of mobility cones,  (\ref{sec:tHaah}) details of the
twisted Haah's code, and (\ref{sec:d_Haah}) the discreteness of fracton self-statistics in the Haah's code and its twisted variant.

\section{counting of windmills for foliated fractons}\label{sec:counting}

In this section, we enumerate the windmills for foliated fractons. There are eight of them: the $yz$ plane has four quadrants for placing $M_2$, and one has two choices of $M_1$ and $M_3$ (along the $xy$ and $zx$ planes respectively) after fixing $M_2$ as depicted in Fig.~3(b) and (c) of the main text. 
Explicitly, they can be denoted as
$\left[xyz\right]$, $\left[x\overline{yz}\right]$, $\left[\overline{x}y\overline{z}\right]$, $\left[\overline{xy}z\right]$, $\left[\overline{xyz}\right]$, $\left[\overline{x}yz\right]$, $\left[x\overline{y}z\right]$, and $\left[xy\overline{z}\right]$. The last four are the spacial inversions of the first four. 

It is important to note that, while windmill processes generally exist (as also observed in our study of Haah's code), the numbers of windmills and independent statistics vary with the underline fracton mobility structure. Our enumeration targets foliated fractons as defined by the main text's Eq.~\eqref{eq:3foliatedfusion}, deferring analogous exhaustive counts of fracton statistical processes in other scenarios (including the fractal cases and many other variants like those in Refs.~\cite{Shirley23, Yan23}) to future studies.

\section{Proof of Eq.~(\ref{eq:FusionExchange})}\label{sec:P}

This section presents a proof for the main text's Eq.~\eqref{eq:FusionExchange}. 

To start, we observe that $\theta_{q\times q'}^{[xyz]}=\theta_{q}^{[xyz]}\theta_{q'}^{[xyz]}$
if $q$ and $q'$ are excitations separated in the $[111]$ direction.
This is because $\theta_{q\times q'}^{[xyz]}$ can be realized as
a composition of the $\left[xyz\right]$ processes for $q$ and $q'$,
which are realized on two parallel windmills that do \emph{not} intersect
with each other. See Fig.~\ref{fig:TopSpin}. 

\begin{figure}[h]
	\includegraphics[width=0.30\columnwidth]{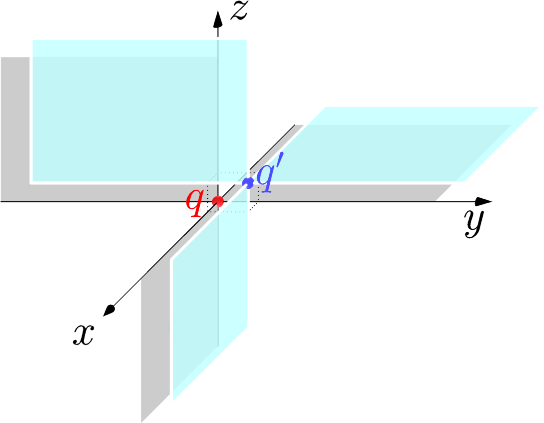}
	
	\caption{Illustration of $\theta_{q\times q'}^{[xyz]}=\theta_{q}^{[xyz]}\theta_{q'}^{[xyz]}$
		for $q$ and $q'$ spatially separated in the $\left[111\right]$
		direction.}
	
	\label{fig:TopSpin} 
\end{figure}

Next, consider the $\left[xyz\right]$ self-statistics of $a\times b\times\overline{\prescript{l\mathbf{m}}{}{a}}\times\overline{\prescript{-l\mathbf{m}}{}{b}}$
for any large enough $l$, where $\mathbf{m}=\mathbf{t}_{x}+\mathbf{t}_{y}+\mathbf{t}_{z}$
denotes the $[111]$ direction. The above observation implies
\begin{equation}
\theta_{a\times b\times\overline{\prescript{l\mathbf{m}}{}{a}}\times\overline{\prescript{-l\mathbf{m}}{}{b}}}^{[xyz]}=\theta_{a\times b}^{[xyz]}\theta_{\overline{\prescript{l\mathbf{m}}{}{a}}}^{[xyz]}\theta_{\overline{\prescript{-l\mathbf{m}}{}{b}}}^{[xyz]}.\label{eq:Sabab}
\end{equation}
Similarly, for later use, we also have 
\begin{equation}
\theta_{a\times\overline{\prescript{l\mathbf{m}}{}{a}}}^{[xyz]}=\theta_{a}^{[xyz]}\theta_{\overline{\prescript{l\mathbf{m}}{}{a}}}^{[xyz]}\qquad\mathrm{and}\qquad\theta_{b\times\overline{\prescript{-l\mathbf{m}}{}{b}}}^{[xyz]}=\theta_{b}^{[xyz]}\theta_{\overline{\prescript{-l\mathbf{m}}{}{b}}}^{[xyz]}.\label{eq:Saabb}
\end{equation}

An alternate expression of the self-statistics of $a\times b\times\overline{\prescript{l\mathbf{m}}{}{a}}\times\overline{\prescript{-l\mathbf{m}}{}{b}}$
can be given in terms of statistics of $a\times\overline{\prescript{l\mathbf{m}}{}{a}}$
and $b\times\overline{\prescript{-l\mathbf{m}}{}{b}}$. To write it
down explicitly, recall that, for each $\mu\in\left\{ x,y,z\right\} $,
\begin{equation}
p_{\mu}^{a}\coloneqq a\times\overline{\prescript{l\mathbf{t}_{\mu}}{}{a}}\qquad\mathrm{and}\qquad p_{\overline{\mu}}^{b}\coloneqq b\times\overline{\prescript{-l\mathbf{t}_{\mu}}{}{b}}\label{eq:dipoles}
\end{equation}
are planons mobile in the directions perpendicular to $\mu$. In particular,
we have $p_{y}^{a}=a\times\overline{\prescript{l\mathbf{t}_{y}}{}{a}}=\prescript{l\mathbf{t}_{x}}{}{a}\times\overline{\prescript{l\mathbf{t}_{x}+l\mathbf{t}_{y}}{}{a}}$
and $p_{z}^{a}=a\times\overline{\prescript{l\mathbf{t}_{z}}{}{a}}=\prescript{l\mathbf{t}_{x}+l\mathbf{t}_{y}}{}{a}\times\overline{\prescript{l\mathbf{m}}{}{a}}$.
Accordingly, 

\begin{equation}
a\times\overline{\prescript{l\mathbf{m}}{}{a}}=a\times\overline{\prescript{l\mathbf{t}_{x}}{}{a}}\times\prescript{l\mathbf{t}_{x}}{}{a}\times\overline{\prescript{l\mathbf{t}_{x}+l\mathbf{t}_{y}}{}{a}}\times\prescript{l\mathbf{t}_{x}+l\mathbf{t}_{y}}{}{a}\times\overline{\prescript{l\mathbf{m}}{}{a}}=p_{x}^{a}\times p_{y}^{a}\times p_{z}^{a}.
\end{equation}
Similarly, $b\times\overline{\prescript{-l\mathbf{m}}{}{b}}=p_{\overline{x}}^{b}\times p_{\overline{y}}^{b}\times p_{\overline{z}}^{b}$.
The associated windmill processes are thus reducible into braidings
of planons. Therefore, $\theta_{a\times\overline{\prescript{l\mathbf{m}}{}{a}}}^{[xyz]}=\theta_{p_{x}^{a}}\theta_{p_{y}^{a}}\theta_{p_{z}^{a}}$,
$\theta_{b\times\overline{\prescript{-l\mathbf{m}}{}{b}}}^{[xyz]}=\theta_{p_{\overline{x}}^{b}}\theta_{p_{\overline{y}}^{b}}\theta_{p_{\overline{z}}^{b}}$,
and more importantly, for our current purpose, $\theta_{a\times b\times\overline{\prescript{l\mathbf{m}}{}{a}}\times\overline{\prescript{-l\mathbf{m}}{}{b}}}^{[xyz]}$
can be expressed as 
\begin{equation}
\theta_{a\times b\times\overline{\prescript{l\mathbf{m}}{}{a}}\times\overline{\prescript{-l\mathbf{m}}{}{b}}}^{[xyz]}=\theta_{a\times\overline{\prescript{l\mathbf{m}}{}{a}}}^{[xyz]}\cdot\theta_{b\times\overline{\prescript{-l\mathbf{m}}{}{b}}}^{[xyz]}\cdot S_{ab}^{x}S_{ab}^{y}S_{ab}^{z},\label{eq:Sab}
\end{equation}
where $\theta_{p_{\mu}^{a}(p_{\overline{\mu}}^{b})}$ is the braiding
self-statistics of $p_{\mu}^{a}$ ($p_{\overline{\mu}}^{b}$) and
$S_{ab}^{\mu}$ has been defined in the main text to be the mutual
statistics between $p_{\mu}^{a}$ and $p_{\overline{\mu}}^{b}$. 

Finally, the desired identity 
\begin{equation}
\theta_{a\times b}^{[xyz]}=\theta_{a}^{[xyz]}\theta_{b}^{[xyz]}S_{ab}^{x}S_{ab}^{y}S_{ab}^{z}\label{eq:6}
\end{equation}
is proven by comparing Eqs.~(\ref{eq:Sabab}), (\ref{eq:Saabb}),
and (\ref{eq:Sab}).

\section{Non-Pauli stabilizer formulation of 2D double semion model\label{sec:TGT}}

There are various ways to represent the 2D double semion phase \cite{Levin2005,Hu2014,TwistedFracton,Dauphinais2019,MagdalenadelaFuente2021,Ellison2021}.
In this section, we turn a twisted gauge model for double semions
into a non-Pauli stabilizer formalism~\cite{nonPauli}, which will
be used in Sec.~\ref{sec:TwistedChecker} to motivate the formulation
of the twist checkerboard models used in the main text.

It is known that the double semion phase can be represented by a twisted
$\mathbb{Z}_{2}$ gauge model \cite{Hu2014,TwistedFracton}. For the
later convenience in making a comparison with layers of a 3D checkerboard,
consider a 2D checkerboard with a physical qubit per vertex and triangulated
as in Fig.~\ref{fig:DS}(a), where the cyan regions (which we call \emph{sites})
are left as holes for holding non-trivial fluxes. The formalism of 2D
twisted gauge models \cite{TwistedFracton} defines, for example,
the flux term $B_{s}$ and the gauge transformation $\tilde{A}_{s}$
at site $s=4$ as
\begin{align}
B_{4} & \coloneqq Z_{13}Z_{14}Z_{36}Z_{46},\\
\tilde{A}_{4} & \coloneqq X_{14}X_{24}X_{46}X_{47}C\left(Z_{13},Z_{14}\right)C\left(Z_{14},-Z_{24}\right)C\left(-Z_{24},Z_{25}\right)\nonumber \\
& \phantom{\coloneqq\;}C\left(Z_{36},Z_{46}\right)C\left(Z_{46},-Z_{47}\right)C\left(-Z_{47},Z_{57}\right)C\left(Z_{24}Z_{25},B_{5}\right),\label{eq:Atwisted}
\end{align}
where each edge (and the qubit on it) is labeled by its end sites,
$B_{5}\coloneqq Z_{24}Z_{25}Z_{47}Z_{57}$ analogous to $B_{4}$,
and 
\begin{equation}
C\left(P,Q\right)\coloneqq\left(-1\right)^{\frac{1}{4}\left(1-P\right)\left(1-Q\right)} \label{eq:C}
\end{equation}
is a shorthand for expressing controlled-$Z$ gates and their generalizations. 

The formalism of twisted gauge models ensures $[\tilde{A}_{s},\tilde{A}_{s'}]=[\tilde{A}_{s},B_{s'}]=[B_{s},B_{s'}]=0$.
In addition, $\tilde{A}_{s}$ corresponds to a projective representation
of $\mathbb{Z}_{2}$ when $B_{s}=-1$. Explicitly, it is straightforward
to check
\begin{equation}
\tilde{A}_{s}^{2}=B_{s}
\end{equation}
from the definition of $\tilde{A}_{s}$ above. Consequently, $(\tilde{A}_{s},B_{s})=\left(1,1\right)$,
$\left(-1,1\right)$, $\left(i,-1\right)$, and $\left(-i,-1\right)$
label the four possible states (with $\left(1,1\right)$ being the
trivial one) at each site $s$. 

\begin{figure}
	\begin{minipage}[c][1\totalheight][t]{0.49\columnwidth}%
		\includegraphics[width=1\columnwidth]{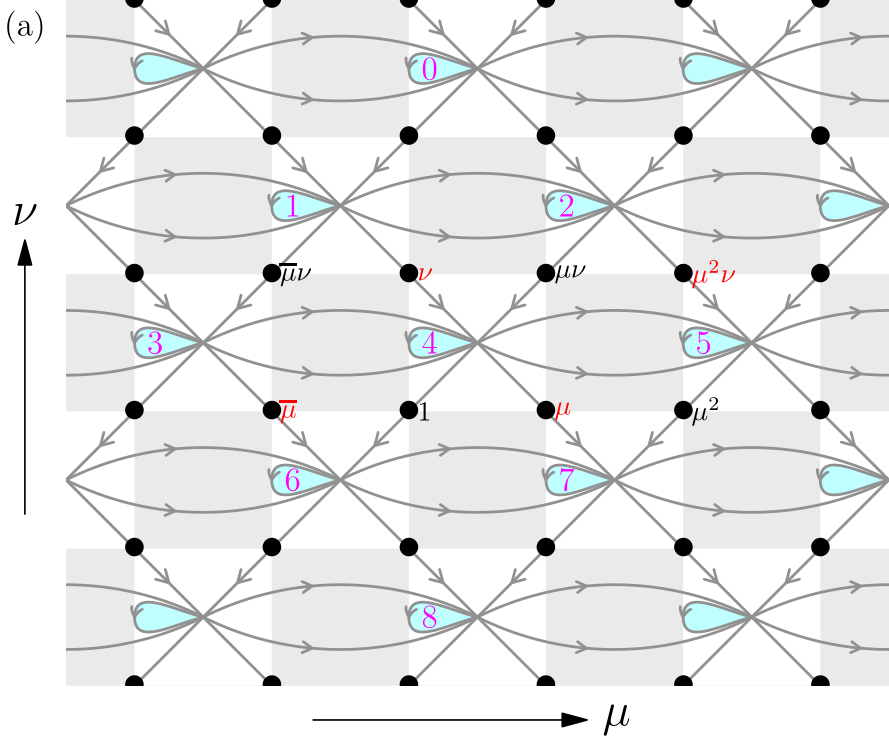}%
	\end{minipage}\qquad{}%
	\begin{minipage}[c][1\totalheight][t]{0.25\columnwidth}%
		\includegraphics[width=1\columnwidth]{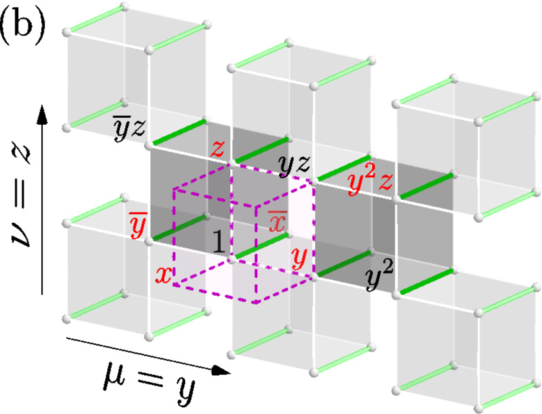}
		
		\medskip{}
		
		\includegraphics[width=1\columnwidth]{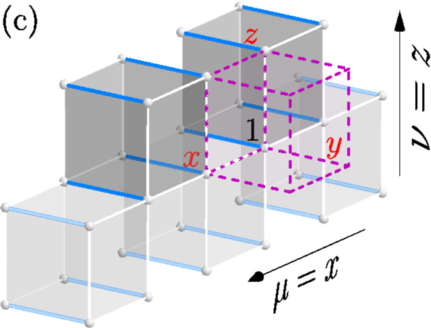}%
	\end{minipage}
	
	\caption{(\textbf{a}) A double semion model on a triangulated 2D checkerboard
		$\mathcal{D}$ with a qubit per vertex (i.e., black dot). Cyan regions
		(which we call \emph{sites}) are punched holes of the plane for holding
		fluxes. Operator $\mathcal{A}_{s}$ for site $s=4$ is supported on
		the eight qubits $\{\overline{\mu},1,\mu,\mu^{2},\overline{\mu}\nu,\nu,\mu\nu,\mu^{2}\nu\}$,
		where $\overline{\mu}\equiv\mu^{-1}$. A label $\mu^{j}\nu^{k}$ is
		in red if $j+k$ is odd. (\textbf{b}) {[}respectively, (\textbf{c}){]}
		A correspondence between an $x$-layer ($y$-layer) of 3D checkerboard
		and the 2D checkerboard $\mathcal{D}$, in which $x$-edges ($y$-edges),
		colored green (blue), correspond to the black dots of $\mathcal{D}$.
		Twisting this layer of the 3D checkerboard model adds, to the $A$
		term at the dashed magenta cube, a phase factor depending on the variables
		$Z_{v\left(1+x\right)}=Z_{v}Z_{vx}$ (respectively, $Z_{v\left(1+y\right)}=Z_{v}Z_{vy}$)
		associated with the $x$-edges ($y$-edges) of the two neighbor cubes
		shown in a darker gray. Vertex $\left(i,j,k\right)$ is labeled $x^{i}y^{j}z^{k}$
		and the label is shown in red (black) if $i+j+k$ is odd (even), where
		$\overline{x}\equiv x^{-1}$ etc. }
	
	\label{fig:DS}
\end{figure}

For convenience, replacing $\tilde{A}_{s}$ with
\begin{equation}
\mathcal{A}_{s}\coloneqq\tilde{A}_{s}(\sqrt{B_{s}})^{\dagger}\label{eq:A_s}
\end{equation}
can turn this double semion model into a \emph{non-Pauli stabilizer
	code}, where $\mathcal{A}_{s}^{2}=B_{s}^{2}=1$, $[\mathcal{A}_{s},\mathcal{A}_{s'}]=[\mathcal{A}_{s},B_{s'}]=[B_{s},B_{s'}]=0$,
and $\mathcal{A}_{s}$ acts identically as $\tilde{A}_{s}$ on the
double semion ground states. The four states at $s$ can thus be labeled
by $\left(\mathcal{A}_{s},B_{s}\right)=\left(1,1\right)$, $\left(-1,1\right)$,
$\left(1,-1\right)$, and $\left(-1,-1\right)$. As shown below, their
self-statistics are $1$, $1$, $i$, and $-i$ respectively.

First, $\left(\mathcal{A}_{s},B_{s}\right)=\left(1,1\right)$ denotes
the absence of excitation and is hence trivial in all aspects (including
self-statistics). Secondly, the self-statistics of $\left(\mathcal{A}_{s},B_{s}\right)=\left(-1,1\right)$
is also trivial; excitations of this type are created in pairs and
moved by purely Pauli $Z$ operators.

As for $\left(\mathcal{A}_{s},B_{s}\right)=\left(1,-1\right)$, let
$\mathcal{S}_{ss'}$ be an operator flipping only $B_{s}$ and $B_{s'}$
at the ends of edge $\left\langle ss'\right\rangle $. Physically,
one may view $\mathcal{S}_{ss'}$ as an elementary string operator
or as a modified bit-flip operator. It takes the form 
\begin{align}
\mathcal{S}_{24} & \coloneqq X_{14}\sqrt{Z_{24}}C\left(Z_{13},Z_{14}\right)C\left(Z_{01}Z_{02},Z_{14}Z_{24}\right)C\left(-Z_{13}Z_{14},Z_{36}Z_{46}\right),\label{eq:S24}\\
\mathcal{S}_{47} & \coloneqq X_{46}(\sqrt{Z_{47}})^{\dagger}C\left(Z_{36},Z_{46}\right)C\left(Z_{13}Z_{14},Z_{36}Z_{46}\right)C\left(-Z_{46}Z_{47},Z_{68}Z_{78}\right)\label{eq:S47}
\end{align}
for $\left\langle ss'\right\rangle =\left\langle 24\right\rangle $
and $\left\langle 47\right\rangle $, and is specified by translation
symmetry in general. By direct computation, one has
\begin{align}
\mathcal{S}_{ss'}^{2}=B_{s}B_{s'} & ,\quad\mathcal{S}_{ss'}^{\dagger}=\mathcal{S}_{ss'}B_{s}B_{s'},\quad\left[\mathcal{S}_{ss'},\mathcal{S}_{s''s'''}\right]=0\;\text{if }s,s',s'',\text{ and }s'''\text{ are four distinct sites,}\label{eq:S_d}\\
\left[\mathcal{S}_{ss'},\mathcal{A}_{s''}\right] & =0,\;\forall s'',\qquad\left[\mathcal{S}_{ss'},B_{s''}\right]=0,\;\forall s''\neq s,s',\qquad\left\{ \mathcal{S}_{ss'},B_{s}\right\} =\left\{ \mathcal{S}_{ss'},B_{s'}\right\} =0,\\
\mathcal{S}_{s's}\mathcal{S}_{s''s} & =i\mathcal{S}_{s''s}\mathcal{S}_{s's},\phantom{-}\;\forall s''\text{ is on the left of }s',\quad\text{e.g.,}\quad\mathcal{S}_{24}\mathcal{S}_{14}=i\mathcal{S}_{14}\mathcal{S}_{24},\label{eq:p1}\\
\mathcal{S}_{ss'}\mathcal{S}_{ss''} & =-i\mathcal{S}_{ss''}\mathcal{S}_{ss'},\;\forall s''\text{ is on the left of }s',\quad\text{e.g.,}\quad\mathcal{S}_{47}\mathcal{S}_{46},=-i\mathcal{S}_{46}\mathcal{S}_{47},\label{eq:p2}\\
\mathcal{S}_{ss'}\mathcal{S}_{s''s} & =-\mathcal{S}_{s''s}\mathcal{S}_{ss'}B_{s},\;\,\phantom{\text{ is on the left of }s'}\qquad\text{ e.g.,}\quad\mathcal{S}_{47}\mathcal{S}_{14}=-\mathcal{S}_{14}\mathcal{S}_{47}B_{4},\label{eq:p3}
\end{align}
where all edges are assumed oriented as in Fig.~\ref{fig:DS}(a).
A generic string operator for $\left(\mathcal{A}_{s},B_{s}\right)=\left(1,-1\right)$
can be written as a product of $\mathcal{S}_{ss'}$ or $\mathcal{S}_{ss'}^{\dagger}$
along a sequence of edges.

To determine the self-statistics of $\left(\mathcal{A}_{s},B_{s}\right)=\left(1,-1\right)$,
we make an counterclockwise exchange of two such excitations, named
$a$ and $b$ for keeping track of their paths. Initialize $a$ and
$b$ at sites $1$ and $2$ respectively in Fig.~\ref{fig:DS}(a).
In terms of the basic moves $M_{1}=\mathcal{S}_{14}\left(4\rightarrow1\right)$,
$M_{2}=\mathcal{S}_{47}^{\dagger}\left(4\rightarrow7\right)$, and
$M_{3}=\mathcal{S}_{24}\left(4\rightarrow2\right)$ for this type
of excitation, the exchange can be realized as follows: first move
$a$ from $1$ to $7$ by $M_{2}M_{1}^{\dagger}=\mathcal{S}_{47}^{\dagger}\mathcal{S}_{14}^{\dagger}$,
then $b$ from $2$ to $1$ by $M_{1}M_{3}^{\dagger}=\mathcal{S}_{14}\mathcal{S}_{24}^{\dagger}$,
and finally $a$ from $7$ to $2$ by $M_{1}M_{2}^{\dagger}=\mathcal{S}_{24}\mathcal{S}_{47}$.
It results in
\begin{equation}
\mathcal{S}_{24}\mathcal{S}_{47}\mathcal{S}_{14}\mathcal{S}_{24}^{\dagger}\mathcal{S}_{47}^{\dagger}\mathcal{S}_{14}^{\dagger}=i.\label{eq:semion}
\end{equation}
Namely, $\left(\mathcal{A}_{s},B_{s}\right)=\left(1,-1\right)$ is
a \emph{semion}. In the derivation of Eq.~\eqref{eq:semion}, we 
used the identities $\mathcal{S}_{24}\mathcal{S}_{47}\mathcal{S}_{14}=-\mathcal{S}_{24}\mathcal{S}_{14}\mathcal{S}_{47}B_{4}=-i\mathcal{S}_{14}\mathcal{S}_{24}\mathcal{S}_{47}B_{4}=i\mathcal{S}_{14}\mathcal{S}_{47}\mathcal{S}_{24}B_{4}B_{4}=i\mathcal{S}_{14}\mathcal{S}_{47}\mathcal{S}_{24}$,
where $\mathcal{S}_{24}$, $\mathcal{S}_{47}$, and $\mathcal{S}_{14}$
are permuted by identities~(\ref{eq:p1}) and (\ref{eq:p3}).

The last excitation type $(\mathcal{A}_{s},B_{s})=\left(-1,-1\right)$
can be viewed as a bound state of $(\mathcal{A}_{s},B_{s})=\left(1,-1\right)$
and $\left(-1,1\right)$. It is created in pairs and moved by $\overline{\mathcal{S}}_{ss'}\coloneqq\mathcal{S}_{ss'}Z_{ss'}$.
Its self-statistics is $-i$, by a computation as in Eq.~(\ref{eq:semion})
but with $\overline{\mathcal{S}}_{ss'}$ used instead.

\section{Twisted checkerboard models and their classification\label{sec:TwistedChecker}}

In this section, we provide some more details about the twisted checkerboard
models used in the main text. 

\subsection{Construction details of twisted checkerboard models\label{subsec:ConstructionTCB}}

In this subsection, we explain how the non-Pauli stabilizer formulation
is obtained for the seven twisted checkerboard models. The construction
is done via the correspondence between the double semion model and
each layer of 3D twisted checkerboards as illustrated in Fig.~\ref{fig:DS},
which also motivates the expression of twisted $X$ operators explicitly
given in Sec.~\ref{subsec:derivationFS}.

On the double semion side, for the model on the 2D checkerboard $\mathcal{D}$
in Fig.~\ref{fig:DS}(a), let $\left(Z|\mathcal{D}\right)$ denote
the configuration of all its qubits in the $Z$ basis. Then $\mathcal{A}_{s}$
in Eq.~(\ref{eq:A_s}) takes the form $\mathcal{A}_{s}=A_{s}\phi_{s}\left(Z|\mathcal{D}\right)$,
where $A_{s}$ is a product of four Pauli $X$ operators, and $\phi_{s}\left(Z|\mathcal{D}\right)$
is a U(1) phase associated with the change of $\left(Z|\mathcal{D}\right)$
made by $A_{s}$. In fact, $\phi_{s}\left(Z|\mathcal{D}\right)$ only
depends on  the $Z$'s associated with a finite number of qubits
near $s$. Similarly, $\mathcal{S}_{ss'}$ in Eqs.~(\ref{eq:S24})
and (\ref{eq:S47}) can be viewed as a modified $X$ operator of the
form $\mathcal{X}_{\sigma}=X_{\sigma}\gamma_{\sigma}\left(Z|\mathcal{D}\right)$
with $\gamma_{\sigma}\left(Z|\mathcal{D}\right)\in\mathrm{U}\left(1\right)$,
where $\sigma\in\mathcal{D}$ labels qubits.

On the 3D checkerboard side, by a $\kappa$-layer we mean a single
layer of cubes arranged on a plane perpendicular to the $\kappa$-direction,
where $\kappa=x,y,z$. The $\kappa$-edges within can be labeled by
\begin{equation}
\mathcal{E}_{n}^{\kappa}\coloneqq\left\{ x^{j_{x}}y^{j_{y}}z^{j_{z}}\left(1+\kappa\right)\,|\,j_{\kappa}=n\right\} \label{eq:layer_edge}
\end{equation}
for $n\in\mathbb{Z}$. Each layer corresponds to a 2D checkerboard,
with $\mathcal{E}_{n}^{\kappa}$ mapped to vertices and gray (blank)
cubes to gray (white) squares. The edge variable on $\ell=v\left(1+\kappa\right)$
is $Z_{\ell}\coloneqq Z_{v}Z_{v\kappa}$. Let $\left(Z|\mathcal{E}_{n}^{\kappa}\right)$
denote the configuration of edge variables $Z_{\ell}$ for $\ell\in\mathcal{E}_{n}^{\kappa}$.
The change of $\left(Z|\mathcal{E}_{n}^{\kappa}\right)$ under $A_{c}$
is like that of $\left(Z|\mathcal{D}\right)$ under some $A_{s}$,
if $c$ is adjacent to $\mathcal{E}_{n}^{\kappa}$. We can thus assign
a phase $\phi_{c}\left(Z|\mathcal{E}_{n}^{\kappa}\right)\in\mathrm{U}\left(1\right)$
to this change as in the double semion model. We may require $\phi_{c}\left(Z|\mathcal{E}_{n}^{\kappa}\right)=1$
if $c$ is not adjacent to $\mathcal{E}_{n}^{\kappa}$, for $\left(Z|\mathcal{E}_{n}^{\kappa}\right)$
is unchanged by $A_{c}$ in this situation. Similarly, for the change
of $\left(Z|\mathcal{E}_{n}^{\kappa}\right)$ under $X_{v}$, there
is a phase $\gamma_{v}\left(Z|\mathcal{E}_{n}^{\kappa}\right)\in\text{U}\left(1\right)$
associated analogously as for $\mathcal{S}_{ss'}$ in Eqs.~(\ref{eq:S24})
and (\ref{eq:S47}).

The twisted checkerboard model $H_{\mathrm{cb}}^{x}$ is obtained
by twisting all $x$-layers, namely,
\begin{equation}
A_{c}\rightarrow A_{c}^{x}=A_{c}\prod_{n}\phi_{c}\left(Z|\mathcal{E}_{n}^{x}\right)\equiv A_{c}\Phi_{c}^{x}\qquad\mathrm{and}\qquad X_{v}\rightarrow X_{v}^{x}=X_{v}\prod_{n}\gamma_{v}\left(Z|\mathcal{E}_{n}^{x}\right)\equiv X_{v}\Gamma_{v}^{x},\label{eq:AX}
\end{equation}
where $\Phi_{c}^{x}\coloneqq\prod_{n}\phi_{c}\left(Z|\mathcal{E}_{n}^{x}\right)$
and $\Gamma_{v}^{x}\coloneqq\prod_{n}\gamma_{v}\left(Z|\mathcal{E}_{n}^{x}\right)$
for short. For each $A_{c}^{x}$ (and $X_{v}^{x}$), only the two
layers adjacent to $c$ (respectively, $v$) contribute a nontrivial
factor. Moreover, $\Phi_{c}^{x}$ and $\Gamma_{v}^{x}$ are local
operators. 
The explicit expressions of $A_{c}^{x}$ and $X_{v}^{x}$
are given in the main text and in Sec.~\ref{subsec:derivationFS}
respectively, obtained based on the correspondence between $\mathcal{E}_{n}^{x}$
and $\mathcal{D}$ in Fig.~\ref{fig:DS}(b). The algebraic properties
of $A_{c}^{x}$ and $X_{c}^{x}$ can be derived analogously as in
Sec.~\ref{sec:TGT}. In particular, one has $(A_{c}^{x})^{2}=1$
and $\left[A_{c}^{x},B_{c'}\right]=\left[A_{c}^{x},X_{v}^{x}\right]=0$,
which enable us to analyze the model $H_{\mathrm{cb}}^{x}$ as a non-Pauli
stabilizer code. 

For concreteness, we illustrate how the expression of $A_{c}^{x}$ in the main text is obtained. On the double semion side, label the qubits by $\mu^{j}\nu^{k}$ 
as in Fig.~\ref{fig:DS}(a). Then, for example, $Z_{13}$ and $Z_{14}$
are denoted $Z_{\overline{\mu}\nu}$ and $Z_{\nu}$ respectively below.
The controlled-$Z$ gate on them can be expressed
as
\begin{equation}
C\left(Z_{13},Z_{14}\right)=C\left(Z_{\overline{\mu}\nu},Z_{\nu}\right)=\left(-1\right)^{n_{\overline{\mu}\nu}^{-}n_{\nu}^{-}},
\end{equation}
where $C\left(\cdot,\cdot\right)$ is defined by Eq.~\eqref{eq:C}
and $n_{v}^{-}\coloneqq\frac{1}{2}\left(1-Z_{v}\right)$. Similarly,
we have $C\left(Z_{14},-Z_{24}\right)=\left(-1\right)^{n_{\nu}^{-}n_{\mu\nu}^{+}}$
with $n_{\mu\nu}^{+}\coloneqq\frac{1}{2}(1+Z_{\mu\nu})$ etc. Thus, $\tilde{A}_{4}$ in Eq.~\eqref{eq:Atwisted} can be
reformulated as
\begin{equation}
\tilde{A}_{4}=A_{4}\left(-1\right)^{n_{\overline{\mu}\nu}^{-}n_{\nu}^{-}+n_{\nu}^{-}n_{\mu\nu}^{+}+n_{\mu\nu}^{+}n_{\mu^{2}\nu}^{-}+n_{\overline{\mu}}^{-}n_{1}^{-}+n_{1}^{-}n_{\mu}^{+}+n_{\mu}^{+}n_{\mu^{2}}^{-}}\left(-1\right)^{n_{\mu\nu\left(1+\mu\right)}^{-}n_{\mu\left(1+\mu\right)\left(1+\nu\right)}^{-}}
\end{equation}
where $n_{\mu\nu\left(1+\mu\right)}^{-}\coloneqq\frac{1}{2}\left(1-Z_{\mu\nu\left(1+\mu\right)}\right)$
with $Z_{\mu\nu\left(1+\mu\right)}\coloneqq Z_{\mu\nu}Z_{\mu^{2}\nu}$ and
in general $n_{\kappa}^{\pm}\coloneqq\frac{1}{2}\left(1\pm Z_{\kappa}\right)$
with $Z_{\kappa}\coloneqq\prod_{v\in\kappa}Z_{v}$ for any finite
set $\kappa$ of qubits. Moreover, $\mathcal{A}_{4}$ in Eq.~\eqref{eq:Atwisted}
takes the form
\begin{equation}
\mathcal{A}_{4}=A_{4}\left(-1\right)^{n_{\overline{\mu}}^{-}n_{1}^{-}+n_{\overline{\mu}\nu}^{-}n_{\nu}^{-}+n_{1}^{-}n_{\mu}^{+}+n_{\nu}^{-}n_{\mu\nu}^{+}+n_{\mu}^{+}n_{\mu^{2}}^{-}+n_{\mu\nu}^{+}n_{\mu^{2}\nu}^{-}}\cdot\left(-1\right)^{n_{\mu\nu\left(1+\mu\right)}^{-}n_{\mu\left(1+\mu\right)\left(1+\nu\right)}^{-}}\cdot i^{-n_{\left(1+\overline{\mu}\right)\left(1+\nu\right)}^{-}}. \label{eq:A4}
\end{equation}
Based on the correspondence between $\mathcal{E}_{n}^{x}$
and $\mathcal{D}$ in Fig.~\ref{fig:DS}(b), $A_{c}^{x}$ in the twisted checkerboard model $H_{\mathrm{cb}}^{x}$
can be expressed as $A_{c}^{x}=A_{c}\Phi_{c}^{x}=A_{c}\phi_{\left(1+x\right)\overline{x}c}\phi_{\left(1+x\right)xc}$,
where $\phi_{\left(1+x\right)\overline{x}c}$ and $\phi_{\left(1+x\right)xc}$
denote the twisting factors associated with the $x$-layers containing
$x$-edges whose vertices are $\ell=\left(1+x\right)\overline{x}c$ and $\ell=\left(1+x\right)xc$
respectively. Explicitly, 
\begin{equation}
\phi_{\ell}\coloneqq\left(-1\right)^{n_{\ell\overline{y}}^{-}n_{\ell}^{-}+n_{\ell\overline{y}z}^{-}n_{\ell z}^{-}+n_{\ell}^{-}n_{\ell y}^{+}+n_{\ell z}^{-}n_{\ell yz}^{+}+n_{\ell y}^{+}n_{\ell y^{2}}^{-}+n_{\ell yz}^{+}n_{\ell y^{2}z}^{-}}\cdot\left(-1\right)^{n_{\ell yz\left(1+y\right)}^{-}n_{\ell y\left(1+y\right)\left(1+z\right)}^{-}}\cdot i^{-n_{\ell\left(1+\overline{y}\right)\left(1+z\right)}^{-}}.
\end{equation}
which is obtained from the twisting factor of Eq.~\eqref{eq:A4}
by replacing $Z_{v}$ (which appears in $n_{v}^{\pm}\coloneqq\frac{1}{2}\left(1\pm Z_{v}\right)$
etc) with edge variable, e.g., $Z_{\ell v}=Z_{\overline{x}cv}Z_{cv}$
for $\ell=\left(1+x\right)\overline{x}c$. As illustrated in Fig.~\ref{fig:DS}(b),
for $c=1$ (the dashed magenta cube), $\phi_{\left(1+x\right)\overline{x}c}$
is associated with the $x$-edges of the two darker gray cubes.

Twisting factors $\Phi_{c}^{y}$, $\Gamma_{v}^{y}$, $\Phi_{c}^{z}$,
and $\Gamma_{v}^{z}$ along $y$-layers and $z$-layers are obtained
from $\Phi_{c}^{x}$ and $\Gamma_{v}^{x}$ by the 3-fold rotation
that permutes $x$, $y$, and $z$ directions. They are used in defining
the remaining twisted checkerboard models mentioned in the main text.
For instance, $H_{\mathrm{cb}}^{xy}$ is obtained by the replacement
$A_{c}\rightarrow A_{c}^{xy}=A_{c}\Phi_{c}^{x}\Phi_{c}^{y}$.

Note that one may obtain different twisting factors if $\mathcal{E}_{n}^{\mu}$
with $\mathcal{D}$ are matched in a different way, but this will
not result in new phases. As an illustration, let $\Phi_{c}^{\overline{y}}$
and $\Gamma_{v}^{\overline{y}}$ denote the twisting factors obtained
from the $\mathcal{E}_{n}^{y}$-$\mathcal{D}$ correspondence in Fig.~\ref{fig:DS}(c),
where $\overline{y}$ refers to $\hat{\mu}\times\hat{\nu}=-\hat{y}$.
Although $\Phi_{c}^{\overline{y}}\neq\Phi_{c}^{y}$, the change of
$\mathcal{E}_{n}^{\mu}$-$\mathcal{D}$ correspondence is equivalent
to a retriangulation of $\mathcal{D}$ and hence, as in 2D \cite{Hu2014},
we may use a Pachner-move-motivated local unitary transformation to
identify the corresponding quantum phases.

\subsection{Two derivations of semionic fracton self-statistics\label{subsec:derivationFS}}

In this subsection, two derivations are given to show that the $B$
excitation of $H_{\mathrm{cb}}^{x}$, denoted $m_{x}$, is a semionic
fracton. 

One approach is to work out the details of the twisted X operator
for generating fractional moves of $B$ excitations explicitly. First,
using Eq.~(\ref{eq:AX}) and the $\mathcal{E}_{n}^{x}$-$\mathcal{D}$
correspondence in Fig.~\ref{fig:DS}(b), we obtain
\begin{equation}
X_{v}^{x}=X_{v}\varphi_{\left(1+\overline{x}\right)v}\varphi_{\left(1+\overline{x}\right)xv},\label{eq:Xx-2}
\end{equation}
where $\left(1+\overline{x}\right)v$ and $\left(1+\overline{x}\right)xv$
correspond to the two $x$-edges connecting to $v$. For $\ell$ ($\tilde{\ell}$)
of the from $\left(1+\overline{x}\right)x^{j}y^{k}z^{l}$ with $j+k+l$
even (odd), 
\begin{align}
\varphi_{\ell} & \coloneqq i^{-n_{\ell y}^{-}}\left(-1\right)^{n_{\ell}^{-}n_{\ell\overline{y}}^{-}+n_{\ell+\ell y}^{+}n_{\ell\left(1+y\right)\overline{z}}^{-}+n_{\ell+\ell\overline{y}}^{-}n_{\ell\left(1+\overline{y}\right)z}^{-}},\label{eq:Xx_0}\\
\varphi_{\tilde{\ell}} & \coloneqq i^{n_{\tilde{\ell}y}^{-}}\left(-1\right)^{n_{\tilde{\ell}}^{-}n_{\tilde{\ell}\overline{y}}^{-}+n_{\tilde{\ell}+\tilde{\ell}\overline{y}}^{+}n_{\tilde{\ell}\left(1+\overline{y}\right)\overline{z}}^{-}+n_{\tilde{\ell}+\tilde{\ell}y}^{-}n_{\tilde{\ell}\left(1+y\right)z}^{-}},\label{eq:Xx_1}
\end{align}
where $n_{\kappa}^{\pm}\coloneqq\frac{1}{2}\left(1\pm Z_{\kappa}\right)$
with $Z_{\kappa}\coloneqq\prod_{v\in\kappa}Z_{v}$ for any finite
set of vertices $\kappa$ (written as a formal sum). The commutation
relation 
\begin{equation}
\!\!X_{v}^{x}X_{v'}^{x}=\begin{cases}
X_{v'}^{x}X_{v}^{x}\prod_{c\ni v,v'}B_{c}, & v^{\prime}v^{-1}\notin xy\text{-plane},\\
\left(-1\right)^{v}iX_{v'}^{x}X_{v}^{x}, & v^{\prime}v^{-1}=xy,x\overline{y},\\
-\left(-1\right)^{v}iX_{v'}^{x}X_{v}^{x}, & v^{\prime}v^{-1}=\overline{x}y,\overline{xy},\\
X_{v'}^{x}X_{v}^{x}, & \mathrm{otherwise}
\end{cases}\label{eq:commutation_Xx}
\end{equation}
can be established by direct computation, where $c\ni v,v'$ labels
colored cubes which contain both $v$ and $v'$, and $\left(-1\right)^{v}\coloneqq\left(-1\right)^{j+k+l}$
for $v=x^{j}y^{k}z^{l}$. In addition, we have $\left(X_{v}^{x}\right)^{2}=\prod_{c\ni v}B_{c}$
and hence $X_{v}^{x\dagger}=X_{v}^{x}\prod_{c\ni v}B_{c}$. One can
use $X_{v}^{x}$ to flip the four $B_{c}$ terms on the cubes connecting
to $v$. This generates fractional moves for windmill processes explicitly.
For $m_{x}$ (i.e., a $B_{c}=-1$ excitation of $H_{\mathrm{cb}}^{x}$),
utilizing Eq.~(\ref{eq:commutation_Xx}), we obtain
\begin{equation}
\theta_{m^{x}}^{[xyz]}=\theta_{m^{x}}^{[x\overline{yz}]}=i\qquad\mathrm{and}\qquad\theta_{m^{x}}^{[\overline{xy}z]}=\theta_{m^{x}}^{[\overline{x}y\overline{z}]}=-i\label{eq:hx1}
\end{equation}
by a straightforward computation.

The second approach provides a fast and intuitive way to write down
fracton self-statistics in twisted models. In this approach, to get
the self-statistics of $m_{x}$ in $H_{\mathrm{cb}}$, we consider
another model $h_{\mathrm{cb}}^{x}$ obtained by twisting the checkerboard
model along $x$-layers only in the $x>0$ region. Let $b_{\pm}$
denote one $B_{c}=-1$ excitation of $h_{\mathrm{cb}}^{x}$ at $\pm\left(l,l,l\right)$
for $l>0$. Since $h_{\mathrm{cb}}^{x}$ is twisted as $H_{\mathrm{cb}}^{x}$
for $x>0$ and untwisted for the $x<0$ region, $b_{+}$ has the same
self-statistics as $m_{x}$ in $H_{\mathrm{cb}}^{x}$ and meanwhile
$b_{-}$ has purely bonsonic self-statistics. Moreover, note that
both $S_{b_{+}b_{-}}^{\mu}$ and $S_{b_{-}b_{+}}^{\mu}$ are trivial
for $\mu=x,y,z$. Hence we get $\theta_{b_{+}\times b_{-}}^{[\tau]}=\theta_{b_{+}}^{[\tau]}\theta_{b_{-}}^{[\tau]}$
by Eq.~(\ref{eq:6}) (i.e., Eq.~\eqref{eq:FusionExchange} of the main text) for $\tau=xyz$,
and analogously also the correspondents for $\tau=\overline{xy}z,x\overline{yz}$,
and $\overline{x}y\overline{z}$. Altogether, we have
\begin{equation}
\theta_{m^{x}}^{[\tau]}=\theta_{b_{+}}^{[\tau]}=\theta_{b_{+}\times b_{-}}^{[\tau]}\label{eq:hx2}
\end{equation}
for $\tau=xyz,\overline{xy}z,x\overline{yz}$, and $\overline{x}y\overline{z}$.
This greatly simplifies the computation of $\theta_{m^{x}}^{[\tau]}$
because $b_{+}\times b_{-}$ is reducible into three dipolar planons
(analogous to those in Eq.~(\ref{eq:dipoles})). In $h_{\mathrm{cb}}^{x}$,
one of the planons is semionic, which implies Eq.~(\ref{eq:hx1})
immediately by using Eq.~(\ref{eq:hx2}). The idea of the second
approach, which simplifies the computation of a fracton's statistics
by combining of twisted and untwisted phases, applies to other twisted
models of foliated fractons.

\subsection{Classification of twisted checkerboard models }

In this subsection, we show how local unitary transformations can
be obtained for establishing the classification, claimed in the main
text, of twisted fraction models. Recall that two gapped models belong
to the same phase if and only if their ground states are related by
a local unitary transformation \cite{LocalUnitary}. 

Let us illustrate this by explaining how to find a local unitary transformation
connecting $H_{\mathrm{cb}}^{xy}$ and $H_{\mathrm{cb}}$. For convenience,
we use the $\mathcal{E}_{n}^{x}$-$\mathcal{D}$ and $\mathcal{E}_{n}^{y}$-$\mathcal{D}$
correspondences in Figs.~\ref{fig:DS}(b) and (c) for twisting construction.
Denote the corresponding model by $H_{\mathrm{cb}}^{x\overline{y}}$.
As discussed in the last paragraph of Sec.~\ref{subsec:ConstructionTCB},
it only differs from $H_{\mathrm{cb}}^{xy}$ by a local unitary transformation.
For $H_{\mathrm{cb}}^{x\overline{y}}$, the twisted $X$ operator
takes the form $X_{v}^{x\overline{y}}=X_{v}\Gamma_{v}^{x}\Gamma_{v}^{\overline{y}}$,
where $\Gamma_{v}^{x}$ and $\Gamma_{v}^{\overline{y}}$ are related
by the reflection that swaps $x$ and $y$. By analogy to the commutation
relation of $X_{v}^{x}$ in Eq.~(\ref{eq:commutation_Xx}), one has
that $X_{v}^{x\overline{y}}$ and $X_{v'}^{x\overline{y}}$ anticommute
for $v'=\left(xy\right)^{\pm1}v$ and commute otherwise. Therefore,
the operator $\mathscr{X}_{v}\coloneqq X_{v}^{x\overline{y}}Z_{xyv}$
has the same algebraic properties as $X_{v}$. Hence there should
exist a local unitary transformation $U$ that maps $X_{v}$ to $\mathscr{X}_{v}$
while keeping $Z_{v}$ fixed. 

Working out the details explicitly, we have
\begin{equation}
\mathscr{X}_{v}\coloneqq X_{v}^{x\overline{y}}Z_{xyv}=X_{v}i^{-\left(-1\right)^{v}\left(n_{\overline{x}yv}^{-}+n_{x\overline{y}v}^{-}\right)}\left(-1\right)^{n_{yv}^{-}n_{xyv}^{-}+n_{xv}^{-}n_{xyv}^{-}+n_{v}^{-}n_{x\overline{y}v}^{-}+n_{\overline{y}v}^{-}n_{xv}^{-}+n_{\overline{x}yv}^{-}n_{v}^{-}+n_{\overline{x}v}^{-}n_{yv}^{-}+n_{\overline{xy}v}^{-}n_{\overline{x}v}^{-}+n_{\overline{xy}v}^{-}n_{\overline{y}v}^{-}}.
\end{equation}
It implies that the transformation $X_{v}\mapsto\mathscr{X}_{v}$
can be realized by the local unitary 
\begin{align}
U & \coloneqq\prod_{v}i^{-\left(-1\right)^{v}n_{\overline{x}yv}^{-}n_{v}^{-}}\left(-1\right)^{n_{\overline{xy}v}^{-}n_{\overline{x}v}^{-}n_{v}^{-}+n_{\overline{xy}v}^{-}n_{\overline{y}v}^{-}n_{v}^{-}}.
\end{align}
Furthermore, by a straightforward but tedious calculation, we can show $UH_{\mathrm{cb}}U^{\dagger}\simeq H_{\mathrm{cb}}^{x\overline{y}}$,
where $\simeq$ means that two operators acting identically on states satisfying $B_{c}=1,\forall c$. Thus, the ground states of $H_{\mathrm{cb}}$, $H_{\mathrm{cb}}^{x\overline{y}}$, and hence also $H_{\mathrm{cb}}^{xy}$ are related by local unitary transformations. They all belong to the same fracton phase. 

Remarkably, by the above discussion, we proved that twistings in two different directions are related and can even cancel each other through a local unitary transformation. For us, this fact was so counterintuitive initially that we did not even think that it could be true until we calculated the fracton self-statistics. 
Combined with this fact, the presence and absence of semionic fracton indeed yields the classification stated in the main text about twisted checkerboard models.

\section{Mobility cones\label{sec:mobility-cones}}

In this section, we provide a mathematical treatment of \emph{mobility
	cones}, together with necessary preliminaries.

\subsection{Mathematical preliminaries}
\begin{defn}
	A set $K\subseteq\mathbb{R}^{d}$ is called a \emph{cone} if  $\alpha v\in K,\forall v\in K,\forall\alpha\in\mathbb{R}_{\geq0}$.
\end{defn}
\begin{defn}
	A cone $K\subseteq\mathbb{R}^{d}$ is called \emph{convex} if $\alpha_{1}v_{1}+\alpha_{2}v_{2}\in K,\forall v_{1},v_{2}\in K,\forall\alpha_{1},\alpha_{2}\in\mathbb{R}_{\geq0}$.
\end{defn}
\begin{defn}
	Given a set $S\subseteq\mathbb{R}^{d}$, the \emph{conical hull }of
	$S$ is denoted $\mathrm{cone}\left(S\right)$ and defined as 
	\begin{equation}
	\mathrm{cone}\left(S\right)\coloneqq\left\{ \sum_{j=1}^{n}\alpha_{j}v_{j}|v_{j}\in S,0\leq\alpha_{j}\in\mathbb{R},j,n\in\mathbb{N}\right\} .
	\end{equation}
	For $S=\emptyset$, $\mathrm{cone}\left(S\right)\coloneqq\left\{ \mathbf{0}\right\} $.
	Each element of $\mathrm{cone}\left(S\right)$ is called a \emph{conical
		combination} of $S$. 
\end{defn}
\begin{rem}
	It is clear that $\mathrm{cone}\left(S\right)$ is a convex cone.
	Actually, $\mathrm{cone}\left(S\right)$ is the smallest convex cone
	containing $S$. 
\end{rem}
The above basic definitions about convex cones may be found in the
mathematical literature, e.g., Ref.~\cite{boyd2004convex}.

\subsection{Definition of mobility cones}

We now define the notion of \emph{mobility cones} for quasiparticles
(especially for fractons) in a gapped phase. Let $\mathbb{E}^{3}$
be the Euclidean space where physical degrees of freedom live. We
view $\mathbb{R}^{3}$ as the group of translations. For $K\subseteq\mathbb{R}^{3}$
and $\mathscr{C}\subseteq\mathbb{E}^{3}$, let 
\begin{equation}
\mathscr{C}+K\coloneqq\left\{ x+v|x\in\mathscr{C},v\in K\right\} \subseteq\mathbb{E}^{3}.
\end{equation}
Let $\mathfrak{q}$ denote a quasiparticle superselection sector (i.e.,
a particle type up to local operations). 
\begin{defn}
	A closed convex cone $K\subseteq\mathbb{R}^{3}$ is called a\emph{
		mobility cone} for $\mathfrak{q}$ (and any quasiparticle in this
	superselection sector) if there exists a finite region $\mathscr{C}\subseteq\mathbb{E}^{3}$
	such that $\mathfrak{q}$ can be realized by excitations supported
	inside $\left(\mathscr{C}+K\right)\backslash\left(-n,n\right)^{3}$,
	$\forall n\in\mathbb{N}$, where $\left(-n,n\right)^{3}$ denotes
	a cube centered at the origin and $\left(\mathscr{C}+K\right)\backslash\left(-n,n\right)^{3}$
	is a region with $\left(-n,n\right)^{3}$ excluded. 
\end{defn}
Physically, this means that quasiparticles in the sector $\mathfrak{q}$
can move (probably in a fractional way) to infinity within the conical
region $\mathscr{O}+K$. 
\begin{example}
	The empty set $K=\emptyset$ is a mobility cone for $\mathfrak{q}$
	if and only if $\mathfrak{q}$ is the trivial superselection sector.
	Note $\mathscr{C}+K=\emptyset$ and $\left(\mathscr{C}+K\right)\backslash\left(-n,n\right)^{3}=\emptyset$
	if $K=\emptyset$. A state with excitations supported inside $\left(\mathscr{C}+K\right)\backslash\left(-n,n\right)^{3}=\emptyset$
	simply contains no excitation at all.
\end{example}
\begin{example}
	Both $\mathrm{cone}\left(\hat{x}\right)$ and $-\mathrm{cone}\left(\hat{x}\right)$
	are mobility cones for a quasiparticle with conventional mobility
	along the $\hat{x}$-direction. 
\end{example}
\begin{example}
	Fractons do not have conventional mobility along any direction. Hence
	they allow only two- or three-dimensional mobility cones. Besides,
	a pair of cones related by inversion, $\pm K$, do not have to be
	both mobility cones for a fracton superselection sector, as we have
	noticed in the Haah's code. 
\end{example}

\subsection{Invariance of mobility cone under local unitary}

Suppose that $H$ is a gapped Hamiltonian and $U$ a local unitary.
Each superselection sector $\mathfrak{q}$ of $H$ is then mapped
by $U$ to a superselection sector $\tilde{\mathfrak{q}}$ of $UHU^{\dagger}$.
If $K$ is a mobility cone for $\mathfrak{q}$, it is clear that $K$
is also a mobility cone for $\tilde{\mathfrak{q}}$.

\subsection{Mobility cones in  Haah's code}

Consider the Haah's code $H_{\mathrm{Haah}}$. Denote the lattice
translation group by $\Lambda\coloneqq\{x^{i}y^{j}z^{k}|i,j,k\in\mathbb{Z}\}$.
Let $R\coloneqq\mathbb{Z}_{2}\Lambda$ be the group ring of $\Lambda$
over $\mathbb{Z}_{2}=\left\{ 0,1\right\} $. The elements of $R$
can be thought as Laurent polynomials in $x,y$, and $z$.

For $H_{\mathrm{Haah}}$, excitation patterns within finite regions
are labeled by $\left(a,b\right)\in R^{2}$ with $a$ and $b$ specifying
the configurations of the $A$-term and $B$-term excitations resepctively.
Two excitation patterns $\left(a,b\right)$ and $\left(a',b'\right)$
belong to the same superselection sector, and we write $\left(a,b\right)\sim\left(a',b'\right)$,
if and only if $a-a'\in I$ and $b-b'\in\overline{I}$, where 
\begin{equation}
I\equiv\left(f_{1},f_{2}\right)_{R}\coloneqq\left\{ a_{1}f_{1}+a_{2}f_{2}|a_{1},a_{2}\in R\right\} \qquad\mathrm{and}\qquad\overline{I}\equiv(\overline{f}_{1},\overline{f}_{2})_{R}\coloneqq\{a_{1}\overline{f}_{1}+a_{2}\overline{f}_{2}|a_{1},a_{2}\in R\}
\end{equation}
are two ideals of $R$, and we remind the reader that $f_1 = 1+x+y+z$ and $f_2=1+xy+xz+yz$. Thus, superselection sectors can be labeled
by $\left(R/I\right)\times(R/\overline{I})$, where $R/I$ and $R/\overline{I}$
denote quotient rings. 

In particular, $\left(a,b\right)\sim\left(0,0\right)$ and represents
the trivial superselection sector if and only if $a\in I$ and $b\in\overline{I}$.
All the \emph{nontrivial} superselection sectors are fractonic and
can be put into three groups: type-$A$, type-$B$, and mixed.
\begin{defn}
	A nontrivial superselection sector $\left(\mathfrak{q}_{A},\mathfrak{q}_{B}\right)\in\left(R/I\right)\times(R/\overline{I})$
	is called \emph{type-$A$} (respectively, \emph{type-$B$}) if $\mathfrak{q}_{B}=0$
	(respectively, $\mathfrak{q}_{A}=0$). It is called \emph{mixed} if
	$\mathfrak{q}_{A}\neq0$ and $\mathfrak{q}_{B}\neq0$. 
\end{defn}
The ideal $I\equiv\left(f_{1},f_{2}\right)_{R}$ describes the fractional
moves of $A$-term excitations. It can also be generated by $f_{1}$
and $g=\left(y+z\right)f_{1}+f_{2}=z^{2}+\left(y+1\right)y+y^{2}+y+1$.
Namely, the ideal $\left(f_{1},g\right)_{R}\coloneqq\left\{ a_{1}f_{1}+a_{2}g|a_{1},a_{2}\in R\right\} $
is identical to $\left(f_{1},f_{2}\right)_{R}$. From $g$, it is
clear that 
\begin{equation}
K_{1}\coloneqq\mathrm{cone}\left(\hat{y},\hat{z}\right),\qquad K_{2}\coloneqq\mathrm{cone}\left(\hat{z}-\hat{y},-\hat{y}\right),\qquad K_{3}\coloneqq\mathrm{cone}\left(-\hat{z},\hat{y}-\hat{z}\right)
\end{equation}
are 2D mobility cones for all type-$A$ superselection sectors. Their
spatial inversions $-K_{1}$, $-K_{2}$, and $-K_{3}$ are mobility
cones for all type-$B$ superselction sectors. See Fig.~\ref{fig:haah-1} of the
main text for a visual. In general, one may also consider 3D mobility
cones, like $\mathrm{cone}\left(\hat{x}+\hat{y},\hat{y}+\hat{z},\hat{z}+\hat{x}\right)$
for type-$A$ superselection sector (as implied by $f_{2}=1+xy+yz+zx$).
Nevertheless, for simplicity, we focus on 2D mobility cones along
the $yz$ plane, which suffice for our current purpose. 

To confirm that type-$A$ and type-$B$ can be distinguished by mobility
cones, we show that none of $-K_{i}$'s (respectively, $K_{i}$'s)
is a mobility cone of any type-$A$ (respectively, type-$B$) superselection
sector. Note that it is possible to permute the three $K_{i}$'s while
keeping $I=\left(f_{1},g\right)_{R}$ invariant using two transformations
(1) a swap $y\leftrightarrow z$ and (2) a mapping given by $x\mapsto x\overline{z}$,
$y\mapsto y\overline{z}$, and $z\mapsto\overline{z}$. Hence we only
need to consider one of $K_{i}$'s (e.g., $K_{3}$). See below. 
\begin{claim}
	The cone $-K_{3}$ is not a mobility cone for any type-$A$ superselection
	sector. 
\end{claim}
\begin{proof}
	Argue by contradiction. Assume that $-K_{3}$ is a mobility cone for
	type-$A$ superselection sector $\left(\mathfrak{q}_{A},0\right)$.
	Then there exists $\xi_{n}\in\mathfrak{q}_{A}$ supported inside $\left(\mathscr{C}-K_{3}\right)\backslash\left(-n,n\right)^{3}$,
	$\forall n\in\mathbb{N}$. Note $-K_{3}=\mathrm{cone}\left(\hat{z},\hat{z}-\hat{y}\right)$.
	By using translation symmetry and the modulo operations of first $f_{1}$
	and then $g$, we can reduce $\xi_{n}$ to be of the form $\xi_{n}=(\xi_{n}^{\prime}z+\xi_{n}^{\prime\prime})z^{n}$
	where $\xi_{n}^{\prime}$ and $\xi_{n}^{\prime\prime}$ are arbitrary Laurent polynomials in $y$. Let $\deg\left(\xi_{n}\right)$ denote the highest
	total degree of $\xi_{n}$, and $\deg_{-}^{y}\left(\xi_{n}\right)$
	denote the lowest degree of $\xi_{n}$ in $y$. For example, $\deg\left(\overline{y}z+1\right)=0$
	and $\deg_{-}^{y}\left(\overline{y}z+1\right)=-1$. Note $\deg\left(\xi_{n}\right)-\deg_{-}^{y}\left(\xi_{n}\right)\geq n$.
	Moreover, $\xi_{n}\equiv\xi_{n-1}\pmod I$ implies $\deg\left(\xi_{n}\right)=\deg\left(\xi_{n-1}\right)$
	and $\deg_{-}^{y}\left(\xi_{n}\right)=\deg_{-}^{y}\left(\xi_{n-1}\right)$,
	which can be checked straightforwardly. Hence, we have $\deg\left(\xi_{0}\right)-\deg_{-}^{y}\left(\xi_{0}\right)=\deg\left(\xi_{n}\right)-\deg_{-}^{y}\left(\xi_{n}\right)\geq n,\forall n\in\mathbb{N}$.
	This contradicts the fact that $\xi_{0}$ is an excitation configuration
	inside a finite region, proving the claim.
\end{proof}
Therefore, $-K_{i}$'s (respectively, $K_{i}$'s) are not mobility
cones for type-$A$ (respectively, type-$B$) superselection sector.
Moreover, $K_{i}$'s and $-K_{i}$'s have to be combined to provide
a mobility cone for each mixed superselection sector. This is how
type-$A$, type-$B$, and mixed superselection sectors are distinguished
in the main text.

\section{Twisted Haah's code\label{sec:tHaah}}

This section presents details of the spin model for twisted Haah's
code \cite{PhysRevResearch.2.023353}. It is explicitly defined by
the Hamiltonian 
\begin{equation}
H_{\mathrm{Haah}}^{F}=-\sum_{\lambda}\left(A_{\lambda}+\tilde{B}_{\lambda}\right),
\end{equation}
which is obtained by replacing the $B$-term of the original Haah's
code $H_{\mathrm{Haah}}$ with
\begin{equation}
\tilde{B}_{\lambda}=\begin{pmatrix}t_{1}f_{2}, & t_{2}f_{1}, & f_{2}, & f_{1}\end{pmatrix},\label{eq:haah_b}
\end{equation}
where $t_{1}=x+y+z+\overline{x}y+\overline{y}z+\overline{z}x$ and
$t_{2}=xy+yz+zx+\overline{x}y+\overline{y}z+\overline{z}x$. 

The twisted model is designed such that it gives rise to an emergent
$H_{\mathrm{Haah}}$-governed fermionic gauge theory. This is mentioned
in the main text and can be understood using the exact correspondence
of degrees of freedom \cite{ShirleyFF}
\begin{equation}
X_{\sigma}\mapsto X_{\sigma}\qquad\mathrm{and}\qquad Z_{\sigma}c_{\sigma}\mapsto Z_{\sigma}T_{\sigma}
\end{equation}
between the gauge theory and the spin model, where $c_{\sigma}$ is
a product of Majorana operators, as defined in the main text, and
$T_{\sigma}$ is a product of Pauli $X$'s which correspond to the
$t$ factors in Eq.~(\ref{eq:haah_b}). In the construction, the
exact form $T_{\sigma}$ and $t$ factors are specified by the requirement
that (anti)commutation relations should be preserved during the correspondence. 

\section{\label{sec:d_Haah}Discreteness of fracton self-statistics in the
	Haah's code and its twisted variant}

In this section, we examine in greater detail the $A$-excitation
self-statistics described in the main text for the Haah's code and
its twisted variant. We show that it allows only two discrete values
$\pm1$ for both phases. 

Consider the exchange process $\theta_{a}=M_{3}^{\dagger}M_{2}M_{1}^{\dagger}M_{3}M_{2}^{\dagger}M_{1}$
of $A$-excitations $a$ (initialized at the origin $x^{0}y^{0}z^{0}\equiv1$)
and $\widehat{a}$ (initialized at $\eta_{2}$) in the Haah's code or its twisted variant,
where $M_{i}:1\rightarrow\eta_{i}$ for $i=1,2,3$ are fractional
moves along the $yz$-plane with $\eta_{i}$ located as in Fig.~\ref{fig:fractalSS}.
Making the exchange twice gives
\begin{equation}
\theta_{a}^{2}=M_{3}^{\dagger}(M_{2}M_{1}^{\dagger})M_{3}M_{2}^{\dagger}(M_{1}M_{3}^{\dagger})M_{2}M_{1}^{\dagger}(M_{3}M_{2}^{\dagger})M_{1},\label{eq:fractalS}
\end{equation}
where parentheses are inserted to notionally separate the moves of
$\widehat{a}$ from those of $a$. Note that $f_{1}=1+x+y+z$ gives
a fractional move $1\rightarrow1+\overline{x}f_{1}=\left(1+y+z\right)\overline{x}$.
It can be iterated to generate fractional moves of each $A$-excitation,
along a cyan tetrahedron illustrated in Fig.~\ref{fig:fractalSS}, from the $yz$-plane
(i.e., $x=0$) to any $x=-n$ plane for $n=1,2,3\cdots$. Thus, we
have $\eta_{i}\rightarrow\eta_{i}\left(\overline{x}+\overline{x}y+\overline{x}z\right)\rightarrow\eta_{i}\left(\overline{x}+\overline{x}y+\overline{x}z\right)^{2}\rightarrow\cdots\rightarrow\eta_{i}\left(\overline{x}+\overline{x}y+\overline{x}z\right)^{n}$.
Together with the inverse of $\eta_{j}\rightarrow\eta_{j}\left(\overline{x}+\overline{x}y+\overline{x}z\right)^{n}$,
we can deform the moves $M_{j}M_{i}^{\dagger}:\eta_{i}\rightarrow1\rightarrow\eta_{j}$
for $\widehat{a}$ into
\begin{equation}
N_{ji}:\;\eta_{i}\rightarrow\eta_{i}\left(\overline{x}+\overline{x}y+\overline{x}z\right)^{n}\rightarrow\left(\overline{x}+\overline{x}y+\overline{x}z\right)^{n}\rightarrow\eta_{j}\left(\overline{x}+\overline{x}y+\overline{x}z\right)^{n}\rightarrow\eta_{j}
\end{equation}
to bypass the origin, where $\eta_{i}\left(\overline{x}+\overline{x}y+\overline{x}z\right)^{n}\rightarrow\left(\overline{x}+\overline{x}y+\overline{x}z\right)^{n}\rightarrow\eta_{j}\left(\overline{x}+\overline{x}y+\overline{x}z\right)^{n}$
is the analogue of $\eta_{i}\rightarrow1\rightarrow\eta_{j}$ but
realized on the $x=-n$ plane. Using the original paths $M_{i}$ for
$a$ and the detoured paths $N_{ji}$ for $\widehat{a}$, we can deform
$\theta_{a}^{2}$ (Eq.~\eqref{eq:fractalS}) into a trivial process.
Thus, $\theta_{a}^{2}=1$ and hence $\theta_{a}=\pm1$. 

Actually, the above argument for $\theta_{a}=\pm1$ holds as long as $a$ allows
a mobility cone out of the plane where $\theta_{a}$ is realized.
The discreteness of $\theta_{a}$ justifies its use in distinguishing
fracton orders.

\begin{figure}
	\includegraphics[width=0.3\columnwidth]{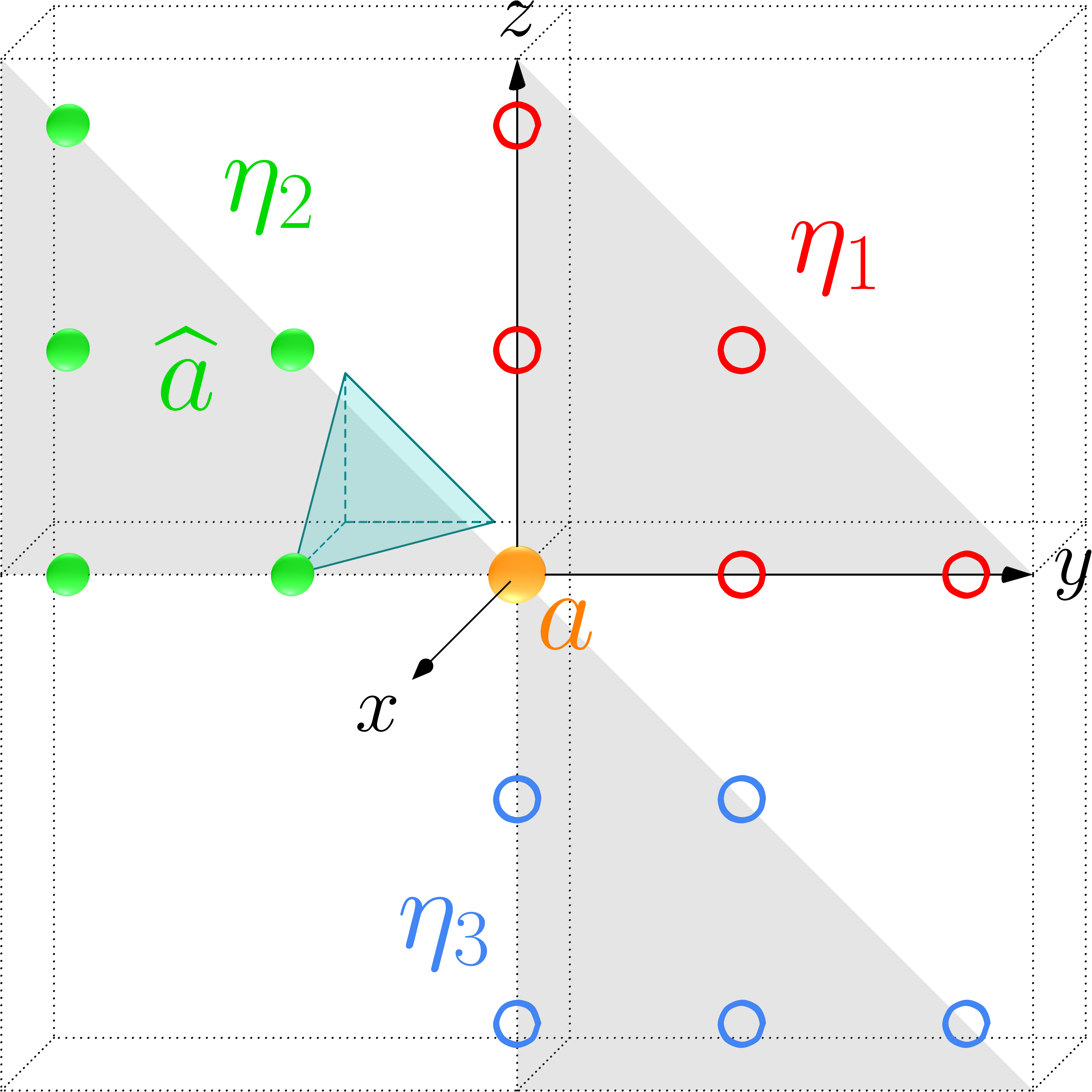}
	
	\caption{A state with $A$-excitations $a$ (at the origin) and $\widehat{a}$
		(at $\eta_{2}$) in the Haah's code or its twisted variant. Each $A$-excitation can be moved along a tetrahedron (like the cyan one) towards where $x=-n$ plane,
		where $n$ can be $1,2,3,\cdots.$}
	
	\label{fig:fractalSS}
\end{figure}	
	
\end{document}